\documentclass[]{fundam}

\publyear{22}
\papernumber{2102}
\volume{185}
\issue{1}

\usepackage{pgfplots}
\pgfplotsset{compat=1.18}

\usepackage{amsfonts}
\usepackage{amsmath}
\usepackage{amssymb}
\usepackage[only,llbracket,rrbracket]{stmaryrd}
\usepackage{algorithm}
\usepackage{algpseudocode}
\usepackage{tikz}
\usepackage{bbm}
\usepackage{mathtools}
\usepackage{hyperref}
\usetikzlibrary {petri,positioning,arrows}
\usepackage{caption}
\usepackage{subcaption}
\usepackage{todonotes}

\newcommand{\positivereals}{\mathbb{R}_{\ge 0}}
\newcommand{\nats}{\mathbb{N}^0}
\newcommand{\net}{\mathcal{N}}
\newcommand{\finally}{F}
\newcommand{\globally}{G}

\newcommand{\fn}[1]{{\sf{#1}}}
\newcommand{\sched}{\fn{scheduled}}
\newcommand{\prob}{\mathbb{P}}
\newcommand{\en}{\fn{en}}
\newcommand*\diff{\mathop{}\!\mathrm{d}}

\title{TAPAAL SMC: Statistical Model Checking of Stochastic Timed-Arc Petri Nets}

\author{Tanguy Dubois \\ Nantes Université \\  \'{E}cole Centrale Nantes,  CNRS, Nantes, France  \and  Kim G. Larsen  \\ Aalborg University\\Department of Computer Science, Aalborg, Denmark \and Ji\v{r}\'{\i} Srba  \\ Aalborg University\\Department of Computer Science, Aalborg, Denmark}

\runninghead{Dubois, Larsen and Srba}{TAPAAL SMC: Statistical Model Checking of Stochastic
Timed-Arc Petri Nets}

\begin{document} 

\maketitle

\begin{abstract}
    Timed-Arc Petri net (TAPN) is a timed extension of the classical Petri net model where tokens have their age and input arcs are associated with  time intervals restricting the ages of tokens available for transition firing.
    Additionally, a TAPN can also contain place invariants constraining the ages of tokens in places, inhibitor arcs preventing a transition from firing and transport arcs that preserve  token ages upon firing.
    This set of features, as much as it allows us to model complex systems, also often makes verification problems computationally hard or even undecidable. Moreover, in order to model real-life examples, additional stochastic aspects are often necessary to capture the desired behaviour.
    We suggest the first stochastic semantics for TAPNs and design and implement the quantitative and qualitative Statistical Model Checking (SMC) algorithms in the model checker TAPAAL.
    We argue  for the semantic choices we made in the stochastic semantics and prove that the semantics is well-behaving. On a number of case studies we demonstrate the practical applicability of our modelling formalism and its SMC implementation. 
\end{abstract}

\sloppy 

\section{Introduction}

In many complex systems, exact deterministic models fail to capture the inherent uncertainty and variability present in real-world processes. Introducing stochastic elements to formal models allows us to better represent these unpredictable dynamics, providing a more realistic framework for analysis~\cite{handbook-stochastic-methods}. However, the exact verification of these models often becomes infeasible due to the sheer computational complexity involved. Instead, statistical model checking~\cite{smc-survey} offers a practical alternative: it provides approximate but computationally efficient results. While statistical methods may not guarantee absolute precision, they offer a valuable trade-off between accuracy and performance, enabling meaningful insights into systems that are otherwise too complex to analyze.

Research on introducing stochastic semantics to Petri nets focuses on incorporating randomness into the modelling of systems with uncertain or probabilistic behaviour. Stochastic Petri nets~\cite{florin_stochastic_1991,bause_stochastic_2013} extend traditional Petri nets by assigning probabilistic firing rates to transitions, enabling the modelling of systems with random delays, such as in queuing or communication networks. This model is further extended to generalized stochastic Petri nets~\cite{Balbo2007} which include also immediate transitions and further extensions with e.g. colored tokens~\cite{colorSPN}. A~complex timing behaviour can 
be modelled by stochastic time Petri nets~\cite{cicirelli_qualitative_2015} where apart from the
stochastic aspects, every transition also contains a firing interval that defines the earliest
and latest firing date of the transition. As there are only finitely many transitions (and hence we need to use only a fixed number of clocks), the SMC engine of UPPAAL is used in~\cite{cicirelli_qualitative_2015} 
to perform statistical model checking. On the other hand, the model of timed-arc Petri nets (TAPN) introduced in~\cite{tapn-intro1,tapn-intro2} (see also an overview paper~\cite{jacobsen_verification_2011}),
requires a new clock for each token in the net and the total number
of clocks cannot be a priori fixed. TAPNs equipped with transport arcs are more expressive 
than timed Petri nets and timed automata~\cite{expressiveness-timedPN,expressiveness-overview},
however, to the best of our knowledge, a stochastic semantics for this model has not been suggested yet. In this paper, we  consider the classical TAPN model extended with
additional features like transport arcs, age invariants and inhibitor arcs as used in the model checker TAPAAL~\cite{DJJJMS:TACAS:12}. We suggest the first stochastic semantics to this extended model, implement a new statistical model checking engine as a part of the TAPAAL suite and evaluate it on
a number of realistic case studies.


We finish the introduction by giving an intuition of our stochastic extension of TAPNs by describing a simple producer/consumer
system depicted as a stochastic TAPN in Figure~\ref{fig:running}. The net in the figure models a producer
that is either in the place \emph{Ready} or \emph{Resting}. As soon as the transition
\emph{Produce} is enabled (e.g. in the initial marking that contains one token of age $0$ in the place \emph{Ready}), we sample its firing date $d \in \positivereals$
from normal distribution with mean $1.5$ and standard deviation $0.3$. The system then delays
for $d$ time units, followed by firing the transition $\emph{Produce}$ which deposits new tokens of age $0$
into the places \emph{Resting} and \emph{Store}. At this moment, two transitions become newly enabled, namely \emph{Recover} and \emph{Consume1}, and we sample their firing dates from uniform distribution between $0.5$ and $1.5$, and exponential distribution with rate $0.4$, respectively. The system now delays until the next 
interesting event occurs, meaning that either one of the two enabled transitions
reaches its firing time, or a new transition becomes enabled.
In the latter case, after three time units the transition \emph{Consume2} becomes enabled as the token's
age now fits into the interval $[3,8]$ and we sample its
firing date from the exponential distribution with rate $0.6$.

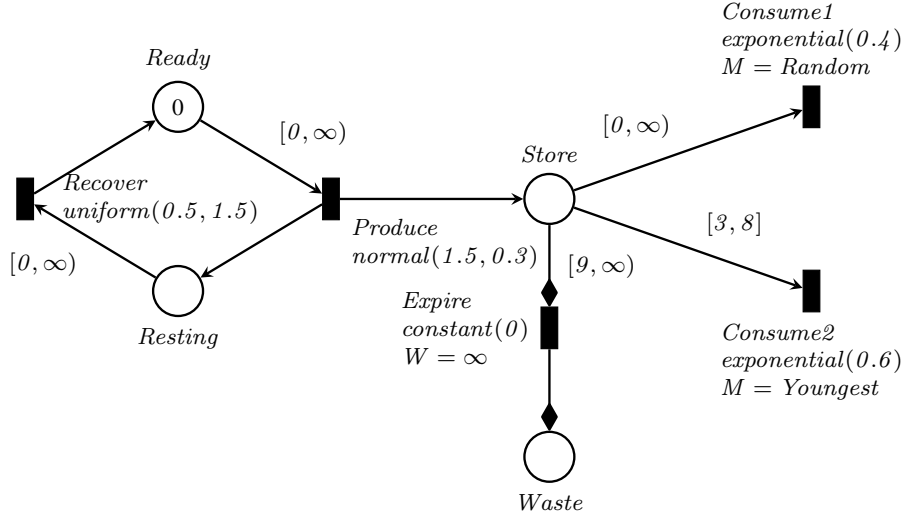
\begin{figure}[t]
\centering
\scalebox{1.1}{
\begin{tikzpicture}[font=\scriptsize, xscale=0.1, yscale=0.07, x=1.66pt, y=2pt]
\tikzstyle{arc}=[->,>=stealth,thick]
\tikzstyle{transportArc}=[->,>=diamond,thick]
\tikzstyle{inhibArc}=[->,>=o,thick]
\tikzstyle{every place}=[minimum size=6mm,thick]
\tikzstyle{every transition} = [fill=black,minimum width=2mm,minimum height=5mm]
\tikzstyle{every token}=[fill=white,text=black]
\tikzstyle{sharedplace}=[place,minimum size=7.5mm,dashed,thin]
\tikzstyle{sharedtransition}=[transition, fill opacity=0, minimum width=3.5mm, minimum height=6.5mm,dashed]
\tikzstyle{urgenttransition}=[place,fill=white,minimum size=2.0mm,thin]
\tikzstyle{uncontrollabletransition}=[transition,fill=white,draw=black,very thick]
\tikzstyle{globalBox} = [draw,thick,align=left]
\node[place, label={[align=left,label distance=0cm]90:$\mathit{Ready}$}] at (540,-360) (Ready) {};
\node at (540.0,-360.0){0};
\node[place, label={[align=left,label distance=0cm]270:$\mathit{Resting}$}] at (540,-810) (Resting) {};
\node[place, label={[align=left,label distance=0cm]90:$\mathit{Store}$}] at (1305,-585) (Store) {};
\node[place, label={[align=left,label distance=0cm]270:$\mathit{Waste}$}] at (1305,-1215) (Waste) {};
\node[transition, label={[align=left,label distance=0.2cm]0:$\mathit{Recover}$\\$\mathit{uniform(0.5,1.5)}$ }] at (225,-585) (Recover) {};
\node[transition, label={[align=left,label distance=0cm]315:$\mathit{Produce}$\\$\mathit{normal(1.5,0.3)}$ }] at (855,-585) (Produce) {};
\node[transition, label={[align=left,label distance=0cm]90:$\mathit{Consume1}$\\$\mathit{exponential(0.4)}$\\$\mathit{M=Random}$ }] at (1845,-360) (Consume1) {};
\node[transition, label={[align=left,label distance=0cm]270:$\mathit{Consume2}$\\$\mathit{exponential(0.6)}$\\$\mathit{M=Youngest}$ }] at (1845,-810) (Consume2) {};
\node[transition, label={[align=left,label distance=0cm]180:$\mathit{Expire}$\\$\mathit{constant(0)}$ 
}] at (1305,-900) (Expire) {};
\draw[arc,pos=0.5] (Ready) to node[bend right=0,auto,align=left] {$\mathit{}$ $\mathit{[0,\infty)}$ } (Produce);
\draw[arc,pos=0.5] (Produce) to node[bend right=0,auto,align=left] {} (Resting);
\draw[arc,pos=0.5] (Resting) to node[bend right=0,auto,align=left] {$\mathit{}$ $\mathit{[0,\infty)}$ } (Recover);
\draw[arc,pos=0.5] (Recover) to node[bend right=0,auto,align=left] {} (Ready);
\draw[arc,pos=0.5] (Produce) to node[bend right=0,auto,align=left] {} (Store);
\draw[arc,pos=0.5] (Store) to node[bend right=0,auto,align=left] {$\mathit{}$ $\mathit{[0,\infty)}$ } (Consume1);
\draw[arc,pos=0.5] (Store) to node[bend right=0,auto,align=left] {$\mathit{}$ $\mathit{[3,8]}$ } (Consume2);
\draw[transportArc,pos=0.5] (Store) to node[bend right=0,auto,align=left] {$\mathit{}$ $\mathit{[9,\infty)}$ } (Expire);
\draw[transportArc,pos=0.5] (Expire) to node[bend right=0,auto,align=left] {$\mathit{}$} (Waste);

\end{tikzpicture}
}
\caption{Stochastic TAPN modelling a produced with two consumers}
\label{fig:running}
\end{figure} 

Now three transitions are enabled and the one with the earliest scheduled firing time will fire, unless another $5$ time units pass and the transition \emph{Consumer2} gets disabled
and its scheduled firing time is deleted. Should no transition be scheduled
for an additional time unit, the transition \emph{Expire} also becomes enabled. Its firing date is sampled
from the constant distribution with mean $0$, meaning that it is scheduled to fire
immediately and it will move the product from \emph{Store} to \emph{Waste}
while preserving its age $9$ (the special arrow types are the transport arcs which move tokens without changing their ages).

Finally, as the place $\emph{Store}$ can possibly contain multiple tokens of different ages,
the firing modes of transitions \emph{Consume1} and \emph{Consume2} specify which tokens
will be consumed during the firing. For \emph{Consume1} we shall use a randomly selected token that enables the transition, however, \emph{Consume2} will always use the youngest available token in the interval $[3,8]$.

We may now ask what is the probability that there appears a product in the 
place \emph{Waste} within the first 20 time units of execution of the system, which
we can, with the help of our SMC tool that executes 461\,121 random runs, estimate to have the probability $0.045 \pm 0.002$ with $95$\% confidence.
This probability will approach the value $1$ as the time horizon gets longer.

\section{Related work}

There are two main approaches for the verification of probabilistic systems. Symbolic probabilistic verification methods \cite{DBLP:books/daglib/0020348}
aim to produce the exact probability of observing a given property by computing or approximating the probabilistic measure. The PRISM tool \cite{kwiatkowska_prism_2011} is known to implement such methods for finite state systems. On the other hand, statistical model-checking (SMC) \cite{smc-survey} generates random runs to produce an estimation of the evaluation of a property. While  symbolic methods can be very costly for exploring the stochastic state space, the latter is much faster at the cost of precision and confidence. However, SMC is known to become costly and imprecise when it comes to verifying rare events, and methods such as importance splitting \cite{jegourel_importance_2013} and importance sampling \cite{jegourel_command-based_2016} have been introduced to solve this issue. 

For stochastic \emph{timed} systems---where both discrete choices and choices of real-time delays are made stochastically---verification is generally undecidable, making SMC indispensable. An exception is the probabilistic semantics for one-clock timed automata proposed in \cite{baier_probabilistic_2007}, where  verification of the very restricted subclass of almost-sure properties was shown to be decidable. 
For networks of general timed automata with an arbitrary number of clocks, a stochastic semantics was given in \cite{DavidLLMPVWFORMATS11}, providing the foundation of UPPAAL's SMC engine \cite{DLLMWCAV11}.
A complementary approach resolves only the discrete choices probabilistically. Recently, Jamroga et. al.~\cite{jamroga2025ptatl} gave such a semantics for networks of probabilistic timed automata with asynchronous execution and imperfect information, together with the logic Probabilistic Timed ATL and complexity results for memoryless strategies.

The first and simplest stochastic extension of Petri nets (sPN) is detailed in \cite{florin_stochastic_1991}. Every transition becomes a stochastic component and samples its firing date according to a given distribution, resulting in a race where the transition with the earliest firing
date wins and fires. This semantics is simple to understand, but only handles transition delays without any timing constraints and every sampled date is considered valid. Verification of sPN can be executed, e.g., using the ORIS tool \cite{paolieri_oris_2021}. 
The basic stochastic Petri nets were extended  with immediate transitions and transitions with deterministic delays into the Deterministic and Stochastic Petri nets (DSPN)~\cite{marsan_petri_1987}, which are implemented in tools like TimeNet 4.0~\cite{timenet4}. Hybrid Petri nets with multiple general transitions (HPnGs)~\cite{pilch_statistical_2017} also support these features, but also allow for continuous transitions. The verification of HPnGs can be done using the tool HYPEG (formerly libhpng)~\cite{pilch_hypeg_2017}. Both essentially use the same method for sampling dates, and thus have the same issues with missing time constraints, however, by mixing stochastic transitions with deterministic timed transitions and continuous transitions, this type of stochastic nets can model more complex timed behaviour. 

\begin{figure}[t]
\begin{center}
\begin{tabular}{|| c | c | c ||} 
 \hline
 Semantics & Time Constraints & Type \\ [0.5ex] 
 \hline\hline
 sPN (ORIS) \cite{bause_stochastic_2013} & No & Strong \\ 
 \hline
 DSPN (TimeNet) \cite{timenet4} & No & Strong \\
 \hline
 HPnGs (HYPEG) \cite{pilch_statistical_2017} & No & Strong \\
 \hline
 sTPN (UPPAAL) \cite{cicirelli_qualitative_2015} & Yes & Strong \\
 \hline
 sTAPN (TAPAAL) & Yes & Weak \\
 \hline
\end{tabular}
\end{center}
\caption{Stochastic timed semantics for Petri nets}
\label{table:semantics summary}
\end{figure}

There exists a stochastic extension of the Time Petri net (TPN) semantics, which achieves even more expressive power by associating a time firing interval to each transition \cite{cicirelli_qualitative_2015}. These nets implement the \emph{strong} semantics~\cite{boyer_comparison_2007} that guarantees that a transition must be fired
within its firing interval and the distributions used for the transitions are essentially scaled to match the firing intervals. An implementation using the UPPAAL SMC engine was initially used to perform SMC on these semantics~\cite{cicirelli_qualitative_2015}. In stochastic TPN time, timing constraints are supported
but are limited to the transitions in the net, which cannot be dynamically changed.
On contrary, timed-arc Petri nets (TAPN)~\cite{tapn-intro1,tapn-intro2}, studied in this paper, associate the timing information
to a potentially unbounded number of tokens in the net and input arcs to transition restrict the
ages of tokens that can be consumed by transition firing. Traditionally, TAPNs rely on the
\emph{weak} semantics meaning that it is allowed to perform time delays that can disable currently
enabled transitions. Further extensions with urgent transitions or age invariants~\cite{jacobsen_verification_2011} are needed to enforce urgent behaviour.
Moreover, an extension of the timed-arc Petri net model with only transport arcs already
makes the model more expressive than other types of timed nets and timed automata~\cite{expressiveness-timedPN,expressiveness-overview}.
No stochastic semantics for TAPNs or other Petri net models that includes the weak (nonurgent)
semantics has been given yet. This is, including an efficient implementation of the framework,
the main contribution of this paper. An overview table of selected semantics and respective tools
is given in Figure~\ref{table:semantics summary}.




\section{Timed-Arc Petri Nets}

We shall first introduce the semantics of timed-arc Petri nets without the stochastic attributes.
Let $\mathcal{I}$ be the set of \emph{timed intervals}
of the from $[a,b]$ where $a \in \nats$, $b \in \nats \cup \{\infty\}$ and $a \leq b$.
Let \emph{time invariants} $\mathcal{I}_\fn{inv} \subseteq \mathcal{I}$ be a subset of
timed intervals of the form above where $a=0$.



\begin{definition}[Timed-arc Petri net]
    A timed-arc Petri net (TAPN) is a~7-tuple 
     $   \mathcal{N}=(P, T, \text{IA}, \text{OA}, \fn{Transport}, \fn{Inhib}, \fn{Inv})
    $    where 
    \begin{itemize}
        \item $P$ is a finite set of places,
        \item $T$ is a finite set of transitions, such that $P \cap T = \emptyset$,
        \item $\text{IA} \subseteq P \times \mathcal{I} \times \mathbb{N} \times T$ is the set of input arcs that connect a place to a transition and are annotated by a time interval and an arc weight, 
        such that if   
            $(p,I,w,t) \in \text{IA}$ and $(p,I',w',t) \in \text{IA}$ then $I = I'$ and $w = w'$,
        \item $\text{OA} \subseteq T \times \mathbb{N} \times P$ is the set of weighted output arcs that connect transitions to places such that 
          if $(t,w,p) \in \text{OA}$ and  $(t,w',p) \in \text{OA}$ then  $w = w'$, 
        \item $\fn{Transport} \subseteq IA \times OA$ is the set of transport arcs
        such that whenever $((p,I,w,t),(t',w',p')) \in \fn{Transport}$ then
            $t = t'$ and $w = w'$, and 
            if $(\alpha,\beta), (\alpha,\beta') \in \fn{Transport}$ then
            $\beta=\beta'$ and symmetrically if $(\alpha,\beta), (\alpha',\beta) \in \fn{Transport}$
            then $\alpha=\alpha'$,
        \item $\fn{Inhib} \subseteq P \times \mathbb{N} \times T$ is the set of weighted inhibitor arcs, and
        \item $\fn{Inv} : P \rightarrow \mathcal{I}_\fn{inv}$ is the function assigning age invariants to places.
    \end{itemize}
\end{definition}

We note that the definition implies that a given place and a transition cannot be connected by both a normal and a transport arc at the same time.
For a transition $t$, we denote by $\fn{Pre}(t) = \{ p \in P \mid (p,I,w,t) \in \text{IA} \}$ the set of input places,
and by $\fn{Post}(t) = \{ p' \in P \mid (t,w',p') \in \text{OA} \}$ the set of output places.
A marking represents the distribution of tokens together with their ages
across the places in the net.

\begin{definition}[Marking]
    A \emph{marking} $M$ is a finite multiset over $P \times \positivereals$ such that $(p,x)$ represents a token of age $x$ in the place $p$.
    A marking M is \emph{valid} if
    $(p,x) \in M$ implies that $x \in \fn{Inv}(p)$
    for every place $p \in P$.
        We denote the set of all  valid markings by $\mathcal{M}$.
\end{definition}

For a marking $M \in \mathcal{M}$, we use the notation $M(p) = \{ x \mid (p,x) \in M \}$ to denote
the multiset of token ages in the place $p$.

\begin{definition}[Enabled transition]
    A transition $t \in T$ is  \emph{enabled} in a~marking $M$ if there exist multisets of tokens $\mathit{In}
    \subseteq \fn{Pre}(t)\times  \positivereals$ and $\mathit{Out} \subseteq \fn{Post}(t)\times \positivereals$
    such that
    \begin{itemize}
        \item $\mathit{In} \subseteq M$, i.e. the $\mathit{In}$ set contains only tokens present in the marking $M$, 
        \item 
        for every $(p,I,w,t) \in \text{IA}$
        and every $x \in \mathit{In}(p)$ we have $x \in I$ and $|\mathit{In}(p)| = w$, i.e.
        each input arc has exactly $w$ tokens in its input place, all of them satisfying its time guard, 
        \item $|M(p)| < w$ for every $(p,w,t) \in \fn{Inhib}$, i.e. none of the inhibitor arcs has enough tokens in its input place to inhibit the firing of $t$,
        \item 
        for every $((p,I,w,t), (t,w,p')) \in \fn{Transport}$ we have 
        $\mathit{In}(p) = \mathit{Out}(p')$
        and
        $x \in \fn{Inv}(p')$ for every   $x \in \mathit{In}(p)$, 
        i.e. all $w$-many tokens in $p$ can be moved (including their ages) to $p'$ and still satisfy the age invariant of $p'$, and
        \item for any $(t,w',p') \in \text{OA}$ that does not appear in $\fn{Transport}$
        we have $|\mathit{Out}(p')|=w'$ and $x=0$ for every $x \in \mathit{Out}(p')$, i.e.
        every output arc that is not a~transport arc must create exactly $w'$ new tokens of age $0$.
    \end{itemize}
\end{definition}

We denote the set of transitions enabled in a marking $M$ by $\en(M)$.

\begin{definition}[Transition firing]
    If $t$ is enabled in a marking $M$ by the multisets of tokens $In$ and $Out$ then $t$ can \emph{fire} and produce a marking
    \[
        M' = (M \setminus In) \cup Out
    \]
    where $\setminus$ and $\cup$ are operations on multisets, 
    and we write $M \xrightarrow{t} M'$. 
\end{definition}

\begin{definition}[Delay]
    Let $M$ be a marking and let $d \in \positivereals$ be a real number.
    We can delay $d$ time units from $M$ if $x+d \in \fn{Inv}(p)$ for all $(p,x) \in M$,
    and write $M \xrightarrow{d}$ to indicate that a delay is possible. If the delay
    is possible, we write $M \xrightarrow{\delta} M[d]$,  where $M[d]$ is a marking defined by
    $M[d] = \{ (p, x + d) \mid (p,x) \in M \}$.
    
\end{definition}

A \emph{run} of a TAPN from an initial marking $M_0$ is an alternating sequence
of time delays and transition firings such that
$$M_0 \xrightarrow{d_0} M_0[d_0] \xrightarrow{t_0} M_1 \xrightarrow{d_1}
M_1[d_1] \xrightarrow{t_1} M_2 \xrightarrow{d_2} M_2[d_2] \xrightarrow{t_2} \ldots$$

A run is \emph{maximal} if it is either infinite or ends in a deadlock, i.e. in a marking
$M$ such that for any possible delay $d$ we have $\en(M[d]) = \emptyset$. By $\fn{runs}(\mathcal{N})$ we denote the set of all maximal runs of $\mathcal{N}$.

%

\subsection{Logic for Reasoning about Runs}

We shall now define a logic that allows us to argue about marking properties.
A \emph{marking property} $\varphi$ is a Boolean combination of atomic propositions
of the form $p \bowtie n$ where $p \in P$, $n \in \nats$ and 
$\bowtie\ \in \{<, \le, =, \ge, >\}$ with the obvious semantics
that a marking $M$ satisfies an atomic property $p \bowtie n$ if and only if
$|M(p)| \bowtie n$, i.e. the number of tokens in a place $p$ satisfies
the constraint imposed by the atomic proposition. This is naturally extended
to Boolean connectives and we write $M \models \varphi$ if a given marking $M$
satisfies the marking property $\varphi$.
Let $\Phi$ be the set of all marking properties. 

An example
of marking property is $p_1 \geq 5 \wedge p_2 =1$ which holds in a marking
that contains at least five tokens in $p_1$ (irrelevant of their ages)
and exactly one token in the place $p_2$. Note that properties that involve
the ages of the tokens can be encoded by introducing ``monitoring'' transitions
with appropriate time intervals and checking their enabledness.

A run of a TAPN satisfies the (reachability) formula $F \varphi$ if
it contains a marking $M$ such that $M \models \varphi$, and
it satisfies the formulate $G \varphi$ if every marking $M$ on the run satisfies
$M \models \varphi$.

\section{Stochastic Timed-Arc Petri Nets}

We shall now define the model of stochastic timed-arc Petri nets (sTAPN).
Let $\mathcal{N} = (P, T, \text{IA}, \text{OA}, \fn{Transport}, \fn{Inhib}, \fn{Inv})$
be a TAPN as defined in the previous section. We extend it with stochastic features by adding three additional functions\footnote{These functions are defined over the domain of real numbers but in the actual implementation are represented as doubles.} assigning (i) density function to transitions (for sampling firing delays), (ii) weights to transition (for the resolution for firing conflicts)
and (iii) firing mode to transitions (for deciding which tokens are used in transition firing).

The function $\fn{density} : T \rightarrow \mathcal{F}(\mathbb{R}, \positivereals)$ assigns a probability density function\footnote{Where for each $t \in T$, the area under the function is one: $\int_{\mathbb{R}} \fn{density}(t)(\theta) d\theta = 1$.} for each transition. It is used to sample firing dates of newly enabled transitions. 
When we sample a delay for a transition $t \in T$ according to $\fn{density}(t)$, we assume that
we choose a random value according to the distribution, except for the case when the sampled value is negative, in which case we return the delay of $0$.

 The function $\fn{weight} : T \rightarrow \positivereals \cup \{ \infty \}$ assigns a weight (priority mass) to each transition. The weights are used in the event of a firing date collision in case several transitions sample the same firing date. Let $T' \subseteq T$
 be the set of transitions that chose the same firing date then the probability of firing $t \in T'$ is given by a weighted uniform choice \[
    \frac{
        \fn{weight}(t)
    }{
        \sum_{t' \in T'} \fn{weight}(t') 
    }
\]
where we postulate that $\frac{\infty}{\infty}=\frac{0}{0}=1$ and $\frac{r}{\infty}=0$ for
any $r \in \positivereals$.
If one of the competing transitions has an infinite weight then it will always win the competition with firing probability $1$, except for the situation where there are several transitions with infinite
weight---in this case we choose uniformly between them. On the other hand, if one of the competing transitions has a zero weight, then it will never be chosen, unless every other competing transition also has a zero weight---in that case we choose again uniformly among the zero-weight transitions.

Finally, the firing mode function $\fn{mode} : T \rightarrow \{ \mathit{Youngest}, \mathit{Oldest}, \mathit{Random} \}$ is a function that determines for each transition
which tokens are consumed when the transition is fired in a given marking.
Assume that $t \in T$ is enabled in $M$ by the multisets of tokens
$\mathit{In}_1, \ldots, \mathit{In}_n$ (enumerating all the possiblities). If $\fn{mode}(t)=\mathit{Youngest}$
then the transition $t$ will be fired using the set $\mathit{In}_i$, $1 \leq i \leq n$,
that minimizes the sum of ages of all tokens in $\mathit{In}_i$ and similarly
if $\fn{mode}(t)=\mathit{Oldest}$ then the sum will be maximized. In case that that there are several
such sets that minimize/maximize the sum, we uniformly choose one such set.
In case
$\fn{mode}(t)=\mathit{Random}$, we select $\mathit{In}_i$ with uniform probability
$i/n$.

Our stochastic TAPN model uses a \textit{single-server policy} \cite{10.1007/11562948_23}, meaning that once a transition fires, it starts its firing process again (if enabled); and an \textit{enabling memory policy} \cite{10.1007/11562948_23}, meaning that a firing date of a transition is stored as long as the transition remains enabled, and the firing date gets forgotten as soon as the transition becomes disabled. Transition firings that both consume and produce a token to some place do not change enabledness of other transitions that also consume from such a place. In other words, transition firing is considered as an atomic event that takes no time.

Figure~\ref{fig:running} shows an example of a stochastic TAPN, demonstrating all the different features of the stochastic extension,
including different density functions based e.g. on normal and exponential distributions, and
weights assigned to transitions, meaning that the transition \emph{Expire} has a firing
priority in case that another of the remaining four transitions (with the default weight $1$) becomes scheduled at the same time point. The stochastic net also shows examples of different firing modes as
the place \emph{Store} can possibly contain several tokens and the transition \emph{Consume1} will
in such case consume a~token of a random age, while \emph{Consume2} will consume a token 
with the smallest available age.




\subsection{Algorithm for Generating Random Runs}

The heart of an SMC algorithm is the generation of random runs from an initial marking $M_0$, until we reach a deadlock, a given time or step (number of transition firings) bound, or a marking that satisfies a given marking property.
The generation of random runs is executed as follows.
\begin{itemize}
    \item  Newly enabled transitions
randomly sample their firing date according to the corresponding density function.
\item 
The net then delays to the next interesting moment where either one of the enabled
transitions is scheduled to fire or the enabledness of transitions changes and their
scheduled firing dates get updated (if a transition becomes disabled, it is unscheduled
and if it becomes newly enabled, we sample according to its density function).
\item If several transitions are scheduled to fire at the same date, we select
the winner according to the function $\fn{weight}$ and fire the transition
using the tokens selected by the function $\fn{mode}$. The winner then resets its firing date.
\item By firing the winning transition, other transitions may become enabled or disabled.
We update their scheduled firing dates accordingly and repeat the whole process until
a marking property is satisfied (in which case we return true) or until we exceed the
given number of steps (transition firings) or the given time horizon. In this case we return false.
\end{itemize}



\begin{algorithm} 
    \caption{$\fn{RandRun}_\mathcal{N}(M_0,\varphi,c,s)$}
    \begin{algorithmic}[1]
        \State \textbf{Input:} Marking $M_0$, property $\varphi$, time bound $c \geq 0$, and steps bound $s \geq 0$
        \State \textbf{Output:} Boolean indicating if a random run generated from $M_0$ contains a marking satisfying $\varphi$ before $c$ time units passed and $k$ transitions were fired
        \smallskip
        \State for all $t \in T$: $\sched(t) \gets$ sample from $\fn{density}(t)$ if $t \in \en(M_0)$, else $\infty$
	\State $\fn{AccumulatedDelay} \gets 0$;  $\fn{AccumulatedSteps} \gets 0$;  $M \gets M_0$
	\While{$\fn{AccumulatedDelay} \le c \text{ and } \fn{AccumulatedSteps} \le s$}
        \State \algorithmicif \ $M \models \varphi$ \algorithmicthen \ \Return $\fn{true}$  
            \State select the smallest positive delay $d$, $0< d \in \positivereals$,  satisfying  \begin{enumerate}
                \item[(i)] $\en(M) \ne \en(M[d])$,  or 
                \item[(ii)] $\exists \delta > 0. \forall \varepsilon \in (0,\delta]. \ \en(M) \ne \en(M[d + \varepsilon])$. \label{line:nextUpperBound}
            \end{enumerate} \label{nextPoint}

             \If {no such $d$ exists and $\sched(t) = \infty$ for all $t \in T$} \label{ifDeadlock}
             \State \Return $\fn{false}$ \Comment{We are in a deadlock} 
             \Else 
             \State $d \gets \min \{\ d, \ \min\limits_{t \in T} \sched(t) \ \}$ \Comment{Select an earliest interesting delay} \label{smallestDelay}
            \EndIf
            
            \State $M \gets M[d]$ \Comment{Advance to the next interesting time point}
            \State $\fn{AccumulatedDelay} \gets \fn{AccumulatedDelay} + d$
            \For{$t \in T$ s.t. $t \in \en(M) \wedge \sched(t)=\infty$} \label{resample} \Comment{Schedule newly enabled transitions} 
            \State sample a value $v$ from $\fn{density}(t)$
            and set $\sched(t)\gets v+d$
            \EndFor
            \State $\fn{candidates} \gets \{ t \in T \mid \sched(t) - d = 0 \}$ \Comment{Transitions that can fire} \label{makeCandidates}
            \If{$\fn{candidates}$ is not empty}
                \State randomly select $t_\fn{fired} \in \fn{candidates}$ using the weight function $\fn{weights}$
                \State let $M \xrightarrow{t_\fn{fired}} M'$ according to $\fn{mode}(t_\fn{fired})$; $M \gets$ $M'$ \label{line:fire}
                \State $\sched(t_\fn{fired}) \gets \infty$ \Comment{The executed transition is unscheduled}
                \State $\fn{AccumulatedSteps} \gets \fn{AccumulatedSteps} + 1$
            \EndIf \label{endFire}
            \For{$t \in T$} \Comment{Update the scheduled firing times of each $t \in T$} \label{updateDatesBegin}
                \If{$t \notin \en(M)$}
                    \State $\sched(t) \gets \infty$
                \Else 
                    \State $\sched(t) \gets \begin{cases}
                        \sched(t) - d &\text{ if } \sched(t) \ne \infty \\
                        \text{sample from } \fn{density}(t) &\text{ otherwise}
                    \end{cases}$
                    \State $\sched(t) \gets 0$ if $\sched(t) < 0$ \label{line:negative check}
                \EndIf \label{updateDatesEnd}
                \If{$\sched(t) > 0 \text{ and } \exists \delta > 0. \forall \varepsilon \in (0,\delta]. \ t \notin \en(M[\varepsilon])$} \label{detectUpper}
                    \State $\sched(t) \gets \infty$ \Comment{Unschedule $t$ as it will get disabled}
                \EndIf
            \EndFor
	\EndWhile
	\State \Return $\fn{false}$
    \end{algorithmic}
\end{algorithm}

The random run generation is formally described in Algorithm~\ref{alg:randomrun}. The algorithm works as follows.
\begin{enumerate}
    \item While the accumulated delay and steps are under the specified bounds, we repeat until reaching a marking satisfying the given marking property, or a deadlock: \begin{enumerate}
        \item We look for the next smallest interesting delay, which is a date where the enabledness of some transition changes (line \ref{nextPoint}). In particular the condition (ii) expresses the fact
        that when an age of a token reaches the upper bound of a time interval on some arc, an arbitrarily
        small delay can disable a currently enabled transition.
        \item If we cannot find such delay, and there are no scheduled transition firings, we reached a~deadlock and terminate (line \ref{ifDeadlock}).
        \item Otherwise, we delay up to the minimum between the smallest interesting delay and the earliest scheduled firing date (line \ref{smallestDelay}). We also sample
        the firing dates for the newly enabled transitions after 
        the delay $d$ at line~\ref{resample} (note that we add $d$ to the sampled dates as the transition enabledness dates are not shifted by $d$ yet).
        \item If more than one transition is scheduled to fire at this date, we randomly choose a winner $t_\fn{fired}$ using the $\fn{weight}$ function, and we fire it by selecting a multiset of consumed tokens according to $\fn{mode}(t_\fn{fired})$ (lines \ref{makeCandidates} to \ref{endFire}).
        \item We update the scheduled firing dates of each transition: disabled transitions are scheduled at $\infty$, the firing date of scheduled transitions is shifted according to the delay and newly enabled transitions are sampled according to their density functions (lines \ref{updateDatesBegin} to \ref{updateDatesEnd}). The check at line~\ref{line:negative check} 
        truncates negative sampled values to zero.
        \item We check if any future-scheduled transition will get disabled after an arbitrary small delay and unschedule such a transition (line \ref{detectUpper}).
    \end{enumerate}
    \label{alg:randomrun}
\end{enumerate}

We shall first point out that we do not need urgent transitions (once an urgent transition is enabled time cannot elapse) in the stochastic semantics, as urgency can be simulated by
Dirac density function which always samples the value $0$ and hence time cannot elapse
as long as at least one such urgent transition is enabled. Contrary to non-stochastic
semantics, age invariants in places cannot be used to enforce urgency. 
Once a token age in a place
reaches the invariant upper bound, the execution of the random run will
end in a deadlock (unless some transition is actually scheduled to fire at that time).

\subsection{Examples of Random Runs}
In order to better understand the details of the random run generation algorithm, we shall
now exemplify it on a number of examples presented in Figure~\ref{fig:net examples}. In these example, the default $[0,\infty)$ time intervals as well as the default weight 1 and the default random firing mode are omitted.

\begin{figure}[t]
     \begin{subfigure}[b]{0.45\textwidth}
        \centering
        \begin{tikzpicture}[font=\scriptsize, xscale=0.3, yscale=0.3, x=0.5pt, y=0.5pt]
\tikzstyle{arc}=[->,>=stealth,thick]
\tikzstyle{transportArc}=[->,>=diamond,thick]
\tikzstyle{inhibArc}=[->,>=o,thick]
\tikzstyle{every place}=[minimum size=6mm,thick]
\tikzstyle{every transition} = [fill=black,minimum width=2mm,minimum height=5mm]
\tikzstyle{every token}=[fill=white,text=black]
\tikzstyle{sharedplace}=[place,minimum size=7.5mm,dashed,thin]
\tikzstyle{sharedtransition}=[transition, fill opacity=0, minimum width=3.5mm, minimum height=6.5mm,dashed]
\tikzstyle{urgenttransition}=[place,fill=white,minimum size=2.0mm,thin]
\tikzstyle{uncontrollabletransition}=[transition,fill=white,draw=black,very thick]
\tikzstyle{globalBox} = [draw,thick,align=left]
\node[place, label={[align=left,label distance=0cm]90:$p_0$}] at (348,-392) (P0) {};
\node at (348.0,-392.0){0};
\node[place, label={[align=left,label distance=0cm]90:$p_1$}] at (1044,-392) (P1) {};
\node[transition, label={[align=left,label distance=0cm]90:$t_0$\\$ \mathit{uniform}(0,5)$}] at (696,-392) (T0) {};
\draw[arc,pos=0.5] (P0) to node[bend right=0,auto,align=left,below] {$\mathit{}$ $\mathit{[3,5]}$ } (T0);
\draw[arc,pos=0.5] (T0) to node[bend right=0,auto,align=left] {} (P1);

\end{tikzpicture}
        \vspace{1cm}
        \caption{Single transition}
        \label{fig:single transition}
     \end{subfigure}
     \hfill
     \begin{subfigure}[b]{0.45\textwidth}
        \centering
        \begin{tikzpicture}[font=\scriptsize, xscale=0.3, yscale=0.3, x=0.5pt, y=0.5pt]
    \tikzstyle{arc}=[->,>=stealth,thick]
    \tikzstyle{transportArc}=[->,>=diamond,thick]
    \tikzstyle{inhibArc}=[->,>=o,thick]
    \tikzstyle{every place}=[minimum size=6mm,thick]
    \tikzstyle{every transition} = [fill=black,minimum width=2mm,minimum height=5mm]
    \tikzstyle{every token}=[fill=white,text=black]
    \tikzstyle{sharedplace}=[place,minimum size=7.5mm,dashed,thin]
    \tikzstyle{sharedtransition}=[transition, fill opacity=0, minimum width=3.5mm, minimum height=6.5mm,dashed]
    \tikzstyle{urgenttransition}=[place,fill=white,minimum size=2.0mm,thin]
    \tikzstyle{uncontrollabletransition}=[transition,fill=white,draw=black,very thick]
    \tikzstyle{globalBox} = [draw,thick,align=left]
    \node[place, label={[align=left,label distance=0cm, xshift=3mm]270:$p_0$\\$\mathit{Inv:}$ $\mathit{\leq}$ $\mathit{30}$ }] at (495,-180) (P0) {};
    \node at (495.0,-180.0){0};
    \node[place, label={[align=left,label distance=0cm]180:$p_1$}] at (180,-450) (P1) {};
    \node[place, label={[align=left,label distance=0cm]0:$p_2$}] at (810,-450) (P2) {};
    \node[transition, label={[align=left,label distance=0cm]90:$t_0$\\$\mathit{uniform}(0,30)$}] at (180,-180) (T0) {};
    \node[transition, label={[align=left,label distance=0cm]90:$t_1$\\$\mathit{uniform}(0,30)$}] at (810,-180) (T1) {};
    \draw[arc,pos=0.5] (P0) to node[bend right=0,auto,align=left] {$\mathit{}$ $\mathit{[0,15]}$ } (T0);
    \draw[arc,pos=0.5] (T0) to node[bend right=0,auto,align=left] {} (P1);
    \draw[arc,pos=0.5] (P0) to node[bend right=0,auto,align=left, below] {$\mathit{}$ $\mathit{[10,35]}$ } (T1);
    \draw[arc,pos=0.5] (T1) to node[bend right=0,auto,align=left] {} (P2);
    
\end{tikzpicture}
        \caption{Simple race}
        \label{fig:simple race}
     \end{subfigure}

     \begin{subfigure}[b]{0.45\textwidth}
        \centering
\begin{tikzpicture}[font=\scriptsize, xscale=0.2, yscale=0.3, x=0.5pt, y=0.5pt]
\tikzstyle{arc}=[->,>=stealth,thick]
\tikzstyle{transportArc}=[->,>=diamond,thick]
\tikzstyle{inhibArc}=[->,>=o,thick]
\tikzstyle{every place}=[minimum size=6mm,thick]
\tikzstyle{every transition} = [fill=black,minimum width=2mm,minimum height=5mm]
\tikzstyle{every token}=[fill=white,text=black]
\tikzstyle{sharedplace}=[place,minimum size=7.5mm,dashed,thin]
\tikzstyle{sharedtransition}=[transition, fill opacity=0, minimum width=3.5mm, minimum height=6.5mm,dashed]
\tikzstyle{urgenttransition}=[place,fill=white,minimum size=2.0mm,thin]
\tikzstyle{uncontrollabletransition}=[transition,fill=white,draw=black,very thick]
\tikzstyle{globalBox} = [draw,thick,align=left]
\node[place, label={[align=left,label distance=0cm]90:$p_0$}] at (765,-495) (P0) {};
\node at (735.0,-515.0){0};
\node at (735.0,-475.0){0};
\node at (795.0,-515.0){0};
\node at (795.0,-475.0){0};
\node[transition, label={[align=left,label distance=0cm]0:$t_1$\\$\mathit{constant(3.0)}$ \\$\mathit{W=1.0}$}] at (1080,-675) (T0) {};
\node[transition, label={[align=left,label distance=0cm]180:$t_2$\\$\mathit{constant(3.0)}$ \\$\mathit{W=0.0}$ }] at (450,-675) (T1) {};
\node[transition, label={[align=left,label distance=0cm]0:$t_0$\\$\mathit{constant(3.0)}$ \\$\mathit{W=4.0}$ }] at (1080,-315) (T2) {};
\node[transition, label={[align=left,label distance=0cm]180:$t_3$\\$\mathit{constant(3.0)}$ \\$\mathit{W=\infty}$ }] at (450,-315) (T3) {};
\draw[arc,pos=0.5] (P0) to node[bend right=0,auto,align=left] {$\mathit{}$  } (T2);
\draw[arc,pos=0.5] (P0) to node[bend right=0,auto,align=left] {$\mathit{}$  } (T0);
\draw[arc,pos=0.5] (P0) to node[bend right=0,auto,align=left] {$\mathit{}$  } (T1);
\draw[arc,pos=0.5] (P0) to node[bend right=0,auto,align=left] {$\mathit{}$ } (T3);
\end{tikzpicture}
        \caption{Date collision}
        \label{fig:collision ex}
     \end{subfigure}
     \hfill
     \begin{subfigure}[b]{0.45\textwidth}
        \centering
\begin{tikzpicture}[font=\scriptsize, xscale=0.4, yscale=0.4, x=0.5pt, y=0.5pt]
\tikzstyle{arc}=[->,>=stealth,thick]
\tikzstyle{transportArc}=[->,>=diamond,thick]
\tikzstyle{inhibArc}=[->,>=o,thick]
\tikzstyle{every place}=[minimum size=6mm,thick]
\tikzstyle{every transition} = [fill=black,minimum width=2mm,minimum height=5mm]
\tikzstyle{every token}=[fill=white,text=black]
\tikzstyle{sharedplace}=[place,minimum size=7.5mm,dashed,thin]
\tikzstyle{sharedtransition}=[transition, fill opacity=0, minimum width=3.5mm, minimum height=6.5mm,dashed]
\tikzstyle{urgenttransition}=[place,fill=white,minimum size=2.0mm,thin]
\tikzstyle{uncontrollabletransition}=[transition,fill=white,draw=black,very thick]
\tikzstyle{globalBox} = [draw,thick,align=left]
\node[place, label={[align=left,label distance=0cm]90:$p_0$}] at (1035,-405) (P0) {};
\node[transition, label={[align=left,label distance=0cm]90:$t_0$\\$\mathit{constant(3.0)}$ }] at (765,-405) (T0) {};
\node[transition, label={[align=left,label distance=0cm]90:$t_1$\\$\mathit{constant(5.0)}$ \\$\mathit{M=\ ?}$ }] at (1305,-405) (T1) {};
\draw[arc,pos=0.5] (T0) to node[bend right=0,auto,align=left] {} (P0);
\draw[arc,pos=0.5] (P0) to node[bend right=0,auto,align=left, below] {$\mathit{}$ $\mathit{[3,\infty)}$ } (T1);

\end{tikzpicture}
        \vspace{0.5cm}
        \caption{Firing mode}
        \label{fig:firing mode ex}
     \end{subfigure}

     \begin{subfigure}[b]{0.45\textwidth}
        \centering
\begin{tikzpicture}[font=\scriptsize, xscale=0.2, yscale=0.3, x=0.5pt, y=0.5pt]
\tikzstyle{arc}=[->,>=stealth,thick]
\tikzstyle{transportArc}=[->,>=diamond,thick]
\tikzstyle{inhibArc}=[->,>=o,thick]
\tikzstyle{every place}=[minimum size=6mm,thick]
\tikzstyle{every transition} = [fill=black,minimum width=2mm,minimum height=5mm]
\tikzstyle{every token}=[fill=white,text=black]
\tikzstyle{sharedplace}=[place,minimum size=7.5mm,dashed,thin]
\tikzstyle{sharedtransition}=[transition, fill opacity=0, minimum width=3.5mm, minimum height=6.5mm,dashed]
\tikzstyle{urgenttransition}=[place,fill=white,minimum size=2.0mm,thin]
\tikzstyle{uncontrollabletransition}=[transition,fill=white,draw=black,very thick]
\tikzstyle{globalBox} = [draw,thick,align=left]
\node[place, label={[align=left,label distance=0cm]180:$\mathit{p_0}$}] at (720,-540) (P0) {};
\node at (720.0,-540.0){0};
\node[place, label={[align=left,label distance=0cm]90:$\mathit{p_1}$}] at (1395,-540) (P1) {};
\node[transition, label={[align=left,label distance=0cm]0:$\mathit{t_0}$\\$\mathit{constant(1.0)}$ \\$\mathit{W=\infty}$ }] at (720,-225) (T0) {};
\node[transition, label={[align=left,label distance=0cm, xshift=5mm]270:$\mathit{t_1}$\\$\mathit{constant(1.0)}$ }] at (1035,-540) (T1) {};
\draw[arc,pos=0.5] (P0)  .. controls(850.0,-382.5) .. (780, -295) to node[bend right=0,auto,align=left] {} (T0);
\draw[arc,pos=0.5] (T0)  .. controls(590.0,-382.5) .. (660,-470) to node[bend right=0,auto,align=left] {} (P0);
\draw[arc,pos=0.5] (P0) to node[bend right=0,auto,align=left] {} (T1);
\draw[arc,pos=0.5] (T1) to node[bend right=0,auto,align=left] {} (P1);

\end{tikzpicture}
        \caption{Atomic firing}
        \label{fig:atimic firing}
     \end{subfigure}
     \hfill
     \begin{subfigure}[b]{0.45\textwidth}
        \centering
\begin{tikzpicture}[font=\scriptsize, xscale=0.2, yscale=0.3, x=0.5pt, y=0.5pt]
\tikzstyle{arc}=[->,>=stealth,thick]
\tikzstyle{transportArc}=[->,>=diamond,thick]
\tikzstyle{inhibArc}=[->,>=o,thick]
\tikzstyle{every place}=[minimum size=6mm,thick]
\tikzstyle{every transition} = [fill=black,minimum width=2mm,minimum height=5mm]
\tikzstyle{every token}=[fill=white,text=black]
\tikzstyle{sharedplace}=[place,minimum size=7.5mm,dashed,thin]
\tikzstyle{sharedtransition}=[transition, fill opacity=0, minimum width=3.5mm, minimum height=6.5mm,dashed]
\tikzstyle{urgenttransition}=[place,fill=white,minimum size=2.0mm,thin]
\tikzstyle{uncontrollabletransition}=[transition,fill=white,draw=black,very thick]
\tikzstyle{globalBox} = [draw,thick,align=left]
\node[place, label={[align=left,label distance=0cm]270:$\mathit{p_0}$}] at (945,-495) (P0) {};
\node at (945.0,-495.0){0};
\node[place, label={[align=left,label distance=0cm]90:$\mathit{p_1}$}] at (585,-135) (P1) {};
\node[place, label={[align=left,label distance=0cm]90:$\mathit{p_2}$}] at (1665,-495) (P2) {};
\node[transition, label={[align=left,label distance=0cm]0:$\mathit{t_0}$\\$\mathit{constant(1.0)}$ \\$\mathit{W=\infty}$ }] at (945,-135) (T0) {};
\node[transition, label={[align=left,label distance=0cm]270:$\mathit{t_2}$\\$\mathit{constant(0.0)}$ }] at (585,-495) (T1) {};
\node[transition, label={[align=left,label distance=0cm]90:$\mathit{t_1}$\\$\mathit{constant(1.0)}$ }] at (1305,-495) (T2) {};
\draw[arc,pos=0.5] (P0) to node[bend right=0,auto,align=left] { } (T0);
\draw[arc,pos=0.5] (T0) to node[bend right=0,auto,align=left] {} (P1);
\draw[arc,pos=0.5] (P1) to node[bend right=0,auto,align=left] {} (T1);
\draw[arc,pos=0.5] (T1) to node[bend right=0,auto,align=left] {} (P0);
\draw[arc,pos=0.5] (P0) to node[bend right=0,auto,align=left] { } (T2);
\draw[arc,pos=0.5] (T2) to node[bend right=0,auto,align=left] {} (P2);

\end{tikzpicture}
        \caption{Sequential firing}
        \label{fig:sequential firing}
     \end{subfigure}
        \caption{Examples of stochastic timed-arc Petri nets}
        \label{fig:net examples}
\end{figure}
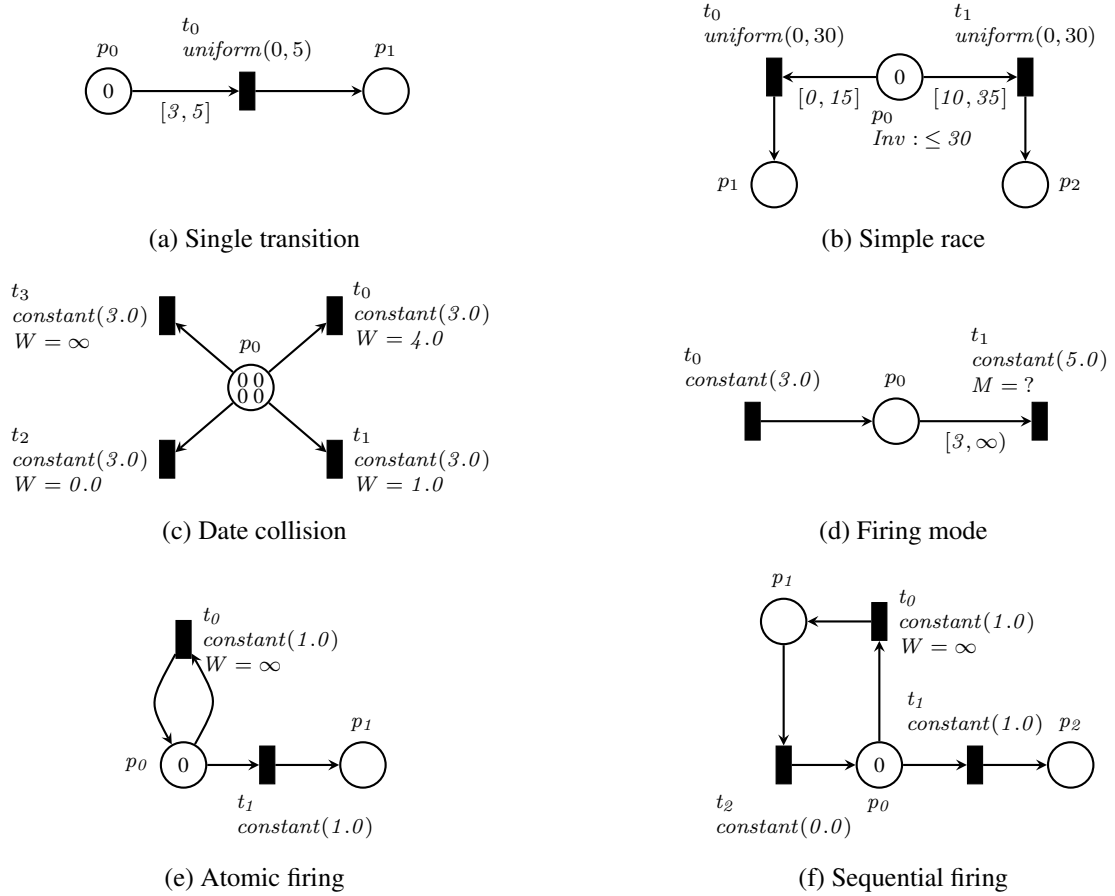


Let us first consider the net from
Figure~\ref{fig:single transition}. Initially, $t_0$ is not enabled so we
will delay to the next interesting event, which is after $3$ time units 
where $t_0$ becomes enabled. At this moment, we sample a~delay $d$ uniformly from
the interval $[0,5]$. If $3+d \leq 5$ then the next interesting delay is
$3+d$ where we fire $t_0$ and reach a deadlock. If $3+d >5$ then the next
interesting delay is $2$ time units and the token in $p_0$ reaches the age $5$.
As $t_0$ is not yet scheduled to fire at this point and after arbitrarily small delay
$t_0$ becomes disabled, we unschedule the firing of $t_0$ as the check at line \ref{detectUpper} of the algorithm succeeds. The run then deadlocks.


A race behaviour between two transitions is shown in Figure~\ref{fig:simple race}.
Initially, $t_0$ is enabled and it chooses a firing date $d_0$ uniformly from the 
interval $[0,30]$. If $d_0 \leq 10$ then $t_0$ will fire as the next interesting delay
is $d_0$ where $t_0$ is the only enabled transition.
If $d_0 > 10$ then at $10$ units the transition $t_1$ will also sample
its firing delay $d_1$ from the interval $[0,30]$. Now if $d_0-10 < d_1$ and $d_0 \leq 15$ then
the transition $t_0$ wins the race. Should $d_0-10 > d_1$ then $t_1$ wins
the race and fires. If none of the transitions fired until 15 time units, $t_0$ gets unscheduled. Depending on the sampled delay $d_1$, either $t_1$ fires provided
that $d_1+10 \leq 30$, or the whole net deadlocks if $d_1+10>30$ due
to the invariant on $p_0$ that blocks any time progress once the token in $p_0$
reaches the age $30$. Note that the~firing date collision of $t_0$ and $t_1$
is unlikely as the probability that the transition delays are sampled so that $d_0 = d_1 + 10$ is zero.


However, in case of constant distributions, collision of firing dates is possible
as demonstrated in Figure~\ref{fig:collision ex}. As there are four tokens of age $0$
in the initial marking, each of the four transitions can fire. The scheduled firing
date of each transition will be $3$ and we need to resolve the probability in 
which order they fire. As the weight of $t_3$ is infinity, it will be always
the first transition to fire. After firing $t_3$, three transitions remain enabled.
Transition $t_0$ will fire next with the probability $4/5$ and transition
$t_1$ will fire with the probability $1/5$. After both $t_0$ and $t_1$ fired,
$t_2$ will fire last as its weight is $0$ and it could not compete with $t_0$ and $t_1$.


The example in Figure~\ref{fig:firing mode ex} shows the influence of the firing mode (depicted by $M = \mathit{mode}$ where $\mathit{mode} \in \{\mathit{Random}, \mathit{Youngest}, \mathit{Oldest} \}$) on executions. The transition $t_0$ produces a token every 3 time units. Once the token in $p_0$ reaches the age $3$, also $t_1$ becomes enabled and will fire after 5 time units. At that point, there will be tokens of ages 8, 5, and 2 in $p_0$, and $t_1$ needs to choose which one of the token ages satisfying the time interval $[3, \infty)$ it consumes during the firing. If $\fn{mode}(t_1) = \mathit{Youngest}$ then $t_1$ selects the token of age 5. If $\fn{mode}(t_1) = \mathit{Oldest}$, it selects the token of age 8. If $\fn{mode}(t_1) = \mathit{Random}$ then $t_1$ uniformly chooses between the tokens of ages $5$ and $8$.
In general, the selection of tokens that a transition consumes according to its firing mode is done at the moment when the transition is fired and not at the moment when it becomes enabled.

Finally, Figures~\ref{fig:atimic firing} and~\ref{fig:sequential firing} demonstrate
a difference when a transition gets unscheduled during the random run generation.
In Figure~\ref{fig:atimic firing} the transitions $t_0$ and $t_1$ get scheduled at time $1$
as they both sample from the constant distribution which always returns $1$.
As the weight of $t_0$ is $\infty$ and the default weight of $t_1$ is $1$,
the transition $t_0$ fires first, resets the age of the token in $p_0$ to $0$,
and samples its next firing date (which is again after $1$ time unit). However,
after the firing $t_0$, the transition $t_1$ is still enabled so its scheduled
firing date does not change and it fires (without any further delay) immediately
after $t_0$. The reader may wonder that during the firing of $t_0$, the transition
$t_1$ got temporarily disabled and it should be resampled too. If such behaviour
is desirable then it can be modelled as shown in Figure~\ref{fig:sequential firing}.
In this case the probability that the transition $t_1$ fires is zero, as it gets always
resampled after any firing of $t_0$ (which always wins the race as its weight is $\infty$).

\subsection{Quantitative and Qualitative SMC Algorithms}

We shall now present two SMC algorithms. The first one estimates the probability
of the event that a~random run satisfies a given marking property (\emph{quantitative estimation}),
and the second one tests
a~hypothesis whether this probability is larger than or equal
to a given constant (\emph{qualitative estimation}). The main reason for introducing
qualitative estimation is that it often requires us to execute a~significantly smaller
number of random runs compared to the quantitative estimation.

\begin{algorithm}[t]
    \caption{$\fn{Quantitative\ Probability\ Estimation}_\mathcal{N}(M_0, \varphi, c, s, \rho, \epsilon)$}
    \label{alg:smc}
    \begin{algorithmic}[1]
        \State \textbf{Input:} Marking $M_0$, property $\varphi$, time bound $c \geq 0$, steps bound $s \geq 0$, confidence $0 < \rho < 1$, precision $0 < \epsilon < 1$
        \State \textbf{Output:} The probability $\pm \epsilon$ that $\varphi$ is satisfied with confidence $\rho$ in no more than $c$ time units and $s$ steps
        \smallskip
        \State $N \gets \frac{\ln(2/(1 - \rho))}{2 \epsilon^2}$ \Comment{The number of runs to execute}
        \State $x \gets 0$
        \For{$i = 1$ to $N$}
            \State $x \gets x + \fn{RandRun}_\mathcal{N}(M_0, \varphi, c, s)$ \Comment{Consider true as 1 and false as 0}
        \EndFor
        \State \Return $\frac{x}{N}$
    \end{algorithmic}
\end{algorithm}

The statistical model checking algorithm presented in Algorithm~\ref{alg:smc} uses the random run generator to perform Monte-Carlo simulations for \emph{quantitative probability estimation} \cite{david_stochastic_2014}. It uses the Chernoff-Hoeffding bound to compute, for a given
precision $\epsilon$ and a confidence level $\rho$,
the number of runs $N$ to be executed in the SMC algorithm.
Now, a confidence-interval estimating process---that 
executes $N$ runs and returns the proportion of $\varphi$-satisfying runs---will have probability larger than $\rho$ of returning a value that is no more than $\epsilon$ away from the  real (unknown) probability of $\varphi$ being satisfied.

\begin{algorithm}[t]
    \caption{$\fn{Qualitative\ Estimation}_\mathcal{N}(M_0, \varphi, c, s, p_t, \delta,  \alpha, \beta)$}
    \label{alg:qual_smc}
    \begin{algorithmic}[1]
        \State \textbf{Input:} Marking $M_0$, property $\varphi$, time bound $c \geq 0$, steps bound $s \geq 0$, the probability bound for the test $0 \leq p_t \leq 1$, indifference region width $\delta$, probability of false positive $0 < \alpha < 1$, probability of false negative $0 < \beta < 1$
        \State \textbf{Output:} A boolean indicating if $\prob(\varphi) \geq p_t$ is satisfied with the probability of false positives at most $\alpha$ and the probability of false negatives  at most $\beta$ in no more than $c$ time units and $s$ steps
        \smallskip
        \State $p_0 \gets p_t + \delta$ and $p_1 \gets p_t - \delta$
        \State $r \gets 0$
        \While{$\fn{true}$}
            \State $x \gets \fn{RandRun}_\mathcal{N}(M_0, \varphi, c, s)$
            \State $r \gets r + x \log\left(\frac{p_1}{p_0}\right) + (1-x)\log\left(\frac{1-p_1}{1-p_0}\right)$ \Comment{Consider true as 1 and false as 0}
            \State If $r \leq \log\left(\frac{\beta}{1-\alpha}\right)$ then \Return $\fn{true}$
            \State If $r \geq \log\left(\frac{1-\beta}{\alpha}\right)$ then \Return $\fn{false}$
        \EndWhile
    \end{algorithmic}
\end{algorithm}

The second statistical model checking algorithm presented in Algorithm~\ref{alg:qual_smc} uses the random run generator to perform a {sequential probability ratio test} for \emph{qualitative hypothesis tests}~\cite{david_stochastic_2014, wald-sequential}. As demonstrated in \cite{wald-sequential}, the algorithm terminates with probability one, 
and returns a Boolean value indicating if the test is verified
under the probability of false positives $\alpha$ (the probability of rejecting the hypothesis $\prob(\phi) \geq p_0$ when it is true) 
and the probability of false negatives $\beta$ (the probability of accepting $\prob(\phi) \leq p_1$ when it is false). The indifference region $[p_1 ; p_0]$ (here centered around $p_t$ with the given indifference region width $\delta$) is the region of probabilities that are not considered relevant to deduce a result for the test.

    \section{Induced Probabilistic Semantics}
    \label{sec:measure-corr}

    The random run generation for a stochastic TAPN $\mathcal{N}$, introduced 
    in Algorithm~\ref{alg:randomrun}, 
    defines an~uncountable set of possible runs and 
    induces a probabilistic measure $\prob_{\mathcal{N}}$ on infinite sets of runs of the $\sigma$-algebra generated by cylinders of runs  of the form $$\pi = I_0 . t_0 . \dots . I_n . t_n \in (\mathcal{I}.T)^*$$ where 
    $I_i \in \mathcal{I}$ is a time interval and $t_i \in T$ is a transition for all $i$, $0 \leq i \leq n$.
    Given a run $\rho = M \xrightarrow{d} M[d] \xrightarrow{t} \rho' \in \fn{runs}(\mathcal{N})$ and a cylinder $\pi = I.t'.\pi'$, we write $\rho \in \pi$ if $d \in I$,  $t = t'$ and $\rho' \in \pi'$; for the base case of the empty cylinder $\epsilon$ we postulate that $\rho \in \varepsilon$ for any run $\rho$. In other words, the prefix of the concrete run must agree on the transition firings with the cylinder and the concrete delays must belong to the intervals in the cylinder.

To define the probability measure for our model, a state in a stochastic TAPN can no longer be reduced to a marking $M$. A state is now a pair $(M,\sched)$ where $M$ is a marking and $\sched : T \rightarrow \positivereals \cup \{ +\infty \}$ is a function mapping transitions to their scheduled firing dates as used in Algorithm~\ref{alg:randomrun}.

Our aim is to define the probability
$\prob_\net((M,\sched), \pi)$ that a random run $\rho$ generated according
to Algorithm~\ref{alg:randomrun} belongs to the cylinder $\pi$, i.e. $\rho \in \pi$.
For the case where the cylinder does not restrict anything ($\pi=\epsilon$), we define
$\prob_\net((M,\sched), \epsilon)=1$. Otherwise, we assume that
$\pi = I.t.\pi'$ for some $\pi' \in (\mathcal{I}.T)^*$. In order to define
$$\prob_\net((M,\sched), I.t.\pi')$$ we distinguish two cases.

{\bf CASE A.} If there is $t \in T$ such that $\sched(t)=0$ and $\sched(t) \not= \infty$
for every $t \in \en(M)$, we are ready to execute transition firing as follows.
We define 
$\prob_\net\big((M,\sched), I.t.\pi'\big) = 0$
if either $0 \not\in I$ or $\sched(t)>0$.
Otherwise, we define
\begin{equation} \label{eq:fire}
    \prob_\net\big((M,\sched), I.t.\pi'\big) = 
    \frac{\fn{weight}(t) \cdot \sum\limits_{(M',\sched') \in \fn{Next}(t)} \frac{\prob_\net\big((M',\sched'), \pi'\big)}{|\fn{Next}(t)|}}{\sum\limits_{t' \in \en(M), \sched(t')=0} \fn{weight}(t') }
\end{equation}
where $\fn{Next}(t)=\{(M', \sched') \mid M \xrightarrow{t} M' 
\text{ and } \sched'(t)=\infty \text{ and } \sched'(t')=\sched(t') 
\text{ for all } t' \in T \smallsetminus \{t\}\}$
returns all the possible markings that can be reached
by firing the transition $t$ (depending on its firing mode)
together with the updated schedule where the currently fired
transition $t$ gets unscheduled.
Equation~(\ref{eq:fire}) assumes that $\fn{weight}(t) \in \positivereals \smallsetminus{0}$. The cases where $\fn{weight}(t)=0$
or $\fn{weight}(t)=\infty$ are defined analogously
(we choose uniformly among the scheduled transitions with
weight $\infty$; the same is true for weight $0$ in case no other transition with
a higher weight is scheduled, otherwise the returned
probability is $0$).

{\bf CASE B.} Let us now consider the case where either 
(i) there is no transition $t \in T$ such that $\sched(t)=0$, or
(ii) there is some $t \in \en(M)$
where $\sched(t)=\infty$.
In case (i) we compute the minimum delay $d$ until some
 transition becomes scheduled to fire or where transition enabledness changes
as follows:
 \[
        d = \min( \ \ 
            \min_{t_i \in \en(M)} \sched(t_i), 
            \ \ \ \inf \{ d \in \positivereals \mid \en(M) \ne \en(M[d]) \}  
        \ \ ) \ .
    \]
This closely follows the computation of minimum delay in Algorithm~\ref{alg:randomrun} at line~\ref{nextPoint}.
If $d=\infty$, i.e. no transition is scheduled to fire and the enabledness of transitions never changes, we reached a deadlock and postulate $\prob_\net((M,\sched), I.t.\pi')=0$. 
In case (ii) we need to schedule some currently enabled transition before we can perform a delay 
and we hence set $d=0$.

For both case (i) and (ii), we define the set
of schedules $\fn{SCH}(d)$ such that
$(M[d], \sched')$ can be reached from
$(M,\sched)$ by a delay of duration $d$ as follows: 
$\sched' \in \fn{SCH}(d)$ if and only if
\begin{itemize}
    \item[a)]  $\sched'(t) \in \positivereals$  for every $t \in T$ where 
    $t \in \en(M[d])$ and 
    $\sched(t) = \infty$,  
    \item[b)] $\sched'(t) = \infty$ for every $t \in T$ where $t \not\in \en(M[d])$, or
$\sched(t) > d$ and at the same time there is $\delta > 0$ such that $t \not\in \en(M[d+\epsilon])$
    for every $\epsilon$, $0 < \epsilon < \delta$, and
    \item[c)] $\sched'(t)=\sched(t) - d$ for every other $t \in T$ not covered by the cases a) and b).
\end{itemize}

We note that in case c) we can safely shift the scheduled time by $d$
for every transition that is enabled, scheduled and either
continues to be enabled also right after the delay $d$ or
it will be scheduled at time $d$. Should a transition that
is scheduled strictly after $d$ time units become disabled
right after the delay $d$, we unschedule it in case b).

Let
\begin{equation} \label{eq:density}
        D({\sched'}) = \prod_{t \in \en(M[d]) \wedge \sched(t)=\infty} \fn{density}(t)(\sched'(t)) 
    \end{equation}
be a density function defined on the set $\fn{SCH}(d)$.
Now the probability measure is defined recursively as follows:
 \begin{equation} \label{eq:delay}
        \prob_\net\big((M,\sched), I.t.\pi'\big) = 
                \int_{s \in \fn{SCH}(d)} D(s) \cdot \prob\big((M[d],s), (I-d).t.\pi'\big) \diff s
    \end{equation}
for the case where $d \in I$, otherwise
$\prob_\net((M,\sched), I.t.\pi') = 0$.
Here $I - d$ for an interval $I=[a,b]$ stands for the interval 
$[\max\{0,a-d\}, b-d]$. Notice that Case B is applicable
only finitely many times in a row
due to the fact that enabledness changes after at least $1$ time unit (unless there
is an earlier scheduled transition firing that is handled by Case A), hence the recursive definition is well-defined.

Finally, we define the probability of a trace starting from the initial marking $M_0$ to belong to the cylinder $\pi$ by 
$\prob_\net(M_0, \pi) =  \prob_\net\big((M_0,\sched_0),\pi\big)$ where
$\sched_0(t)=\infty$ for all $t \in T$.

Observing that $\prob_\net((M,\sched), \pi)$ decreases when extending the cylinder $\pi$, it follows from classical concepts of probability theory (see e.g. \cite{ash2000}) that $\prob_\net$ extends uniquely to the smallest $\sigma$-algebra, $\mathfrak{S}^\net$, generated by the above cylinders.


\begin{theorem} \label{theo:prob_measure}
    Given a stochastic TAPN $\net$, $\prob_\net$ is a probability measure over the $\sigma$-algebra $\mathfrak{S}^\net$ generated by the set of cylinders of the form $(\mathcal{I}.T)^*$.
\end{theorem}
\begin{proof}
Let $\pi \in (\mathcal{I}.T)^*$ be a cylinder.
    By induction on $|\pi|$, we prove that for all states $(M,\sched)$ 
    the probability $\prob_\mathcal{N}((M,\sched),\pi) \in [0,1]$ and 
    for any extended cylinder $\pi.I.t$ with an additional interval and transition
    we have $\prob_\mathcal{N}((M,\sched),\pi.T.t) \le \prob_\mathcal{N}((M,\sched),\pi)$.
    
    For the base case $|\pi| = 0$ the property clearly holds. 
    For the inductive case, let $\pi = I.t.\pi'$.
    There are two cases according to the nature of the next delay:
    \begin{itemize}
        \item If there is $t \in T$ such that $\sched(t)=0$ and $\sched(t) \not= \infty$
for every $t \in \en(M)$, we apply Equation~\ref{eq:fire}.
By the induction hypothesis we know that
$\prob_\mathcal{N}((M',\sched'),\pi') \in [0,1]$ which implies
that also $\prob_\mathcal{N}((M,\sched),\pi) \in [0,1]$ as the contributions
of $\prob_\mathcal{N}((M',\sched'),\pi')$ are normalized by the number
of successor states $|\fn{Next}(t)|$ and because 
$\fn{weight}(t) \leq \sum\limits_{t' \in \en(M), \sched(t')=0} \fn{weight}(t')$.
Clearly, 
extending $\pi$ with an additional interval and transition firing cannot
increase the probability as $\prob_\mathcal{N}((M',\sched'),\epsilon) = 1$.
\item In the other case where either 
(i) there is no transition $t \in T$ such that $\sched(t)=0$, or
(ii) there is some $t \in \en(M)$
where $\sched(t)=\infty$, we apply Equation~\ref{eq:delay} 
that is shifting the interval $I$ by some delay $d$ into
$I - d$. As we already argued, this shifting before a transition
firing becomes enabled can happen only finitely many times and we apply
here an inner induction on the number of shifting steps needed before
a transition firing becomes enabled. If no more shifting is required, 
the property holds as argued in the previous case (transition firing).
Otherwise, we by the inner induction hypothesis
assume that $\prob\big((M[d],s), (I-d).t.\pi'\big) \in [0,1]$.
We also observe that $D$ as defined by Equation~\ref{eq:density} over the 
schedules in $\fn{SCH}(d)$ is a density
function and satisfies that $\int_{s \in \fn{SCH}(d)} D(s) \diff s = 1$.
This together with the inner induction hypothesis implies 
that $\prob_\mathcal{N}((M,\sched),\pi) \in [0,1]$.
As in the previous case, extending $\pi$ to a longer cylinder cannot increase
the probability.
    \end{itemize} 
    
    We have now defined a way to assign a probability to each set of runs described by a cylinder. As there is a countable number of cylinders, there exists an unique way to expand this probability to the $\sigma$-algebra of runs generated by such cylinders.  
\end{proof}

Equally important fact is that our properties $\finally \varphi$ and $\globally \varphi$ define measurable sets of runs, thus having well-defined probabilities as postulated by the next theorem.

\begin{theorem}
    Given a marking property $\varphi \in \Phi$, the formulae $\finally \varphi$ and $\globally \varphi$ describe sets of runs belonging to the $\sigma$-algebra $\mathfrak{S}^\net$, thus having well-defined probabilities $\prob_\mathcal{N}(\finally \varphi)$ and $\prob_\mathcal{N}(\globally \varphi)$.
\end{theorem}
\begin{proof}
Let $\varphi$ be a marking property and we shall argue that the set of all runs
that satisfy $F \varphi$ can be described as a countable union of cylinders
and hence such a set has a well-defined probabilistic measure. To do so, we notice that
for any run $\rho$ that belongs to a cylinder
$$\pi = I_0 . t_0 . \dots . I_n . t_n \in (\mathcal{I}.T)^*$$
the actual delays in the run $\rho$ do not change the number of
tokens in the places in any marking that is reachable from the initial marking 
$M_0$ via the sequence of transition firings $t_0 t_1 \ldots t_n$. 
This leads us to the definition of the set $C_\varphi$ of all cylinders
$I_0 . t_0 . \dots . I_n . t_n$ that satisfy that
there is a run
$$M_0 \xrightarrow{d_0} M_0[d_0] \xrightarrow{t_0} M_1 \xrightarrow{d_1}
M_1[d_1] \xrightarrow{t_1} M_2 \xrightarrow{d_2} M_2[d_2] \xrightarrow{t_2} \ldots
\xrightarrow{t_n} M_{n+1}$$
from the initial marking for some delays $d_0, \ldots, d_n$ such that
$M_{n+1} \models \varphi$ and $M_i \not\models \varphi$ for every
$i$, $0 \leq 1 \leq n$.
This means that any run that belongs to any cylinder from $C_\varphi$
satisfies the formula $F \varphi$ and we can define 
$$\prob_\net(\finally \varphi) = \prob_\net( \bigcup_{\pi \in C_\varphi} 
\pi )$$
as a probabilistic measure of this countable union of cylinders that 
contains exactly all runs satisfying $F \varphi$.
For the formula $G \varphi$, we define the probability as $\prob_\net(\globally \varphi) = 1 - \prob_\net(\finally \neg \varphi)$.
\end{proof}


The last theorem proves that the computed probabilities indeed make sense, by exhibiting an equivalence with classical TAPNs. 

\begin{theorem} \label{theorem:equivalence}
Let $\mathcal{N}$ be a stochastic TAPN and let $\mathcal{N}_d$ be a normal TAPN net obtained by removing
all stochastic features in $\mathcal{N}$. Let $\varphi$ be a marking property.
    \begin{itemize} 
   
   \item[(i)]\label{eq:ef} If there is no reachable marking in  $\mathcal{N}_d$  that satisfies $\varphi$ then
   $\prob_\mathcal{N}(\finally \varphi) = 0$.
   
    \item[(ii)] \label{eq:ag}
    If all reachable markings in $\mathcal{N}_d$  satisfy $\varphi$ then 
    $\prob_\mathcal{N}(\globally \varphi) = 1$.
    \end{itemize}
\end{theorem}
\begin{proof}
For the case (i), as $\varphi$ never holds in any reachable marking and the SMC algorithm
is a conservative extension of reachability on classical TAPNs, there is no SMC run that
satisfies $F \varphi$, implying that the probability $\prob_\mathcal{N}(\finally \varphi)$
corresponds to the measure of the empty set, which is clearly equal to 0.
Similarly, in the case (ii), any marking obtained during the SMC run always satisfies
$\varphi$ and hence $\prob_\mathcal{N}(\globally \varphi)$ is the measure of any
possible run, which is clearly equal to $1$. 
\end{proof}

\section{Implementation in TAPAAL and Case Studies}

The SMC run generation is implemented in C++ as part of the {\tt verifydtapn} engine~\cite{JLSST:NFM:14}
and accepts a PNML description of stochastic TAPN together with the query and other
SMC parameters and it executes Algorithm~\ref{alg:smc} and returns the
computed probability as well as other useful statistics like average run length and duration,
number of satisfying and violating runs and data for plotting cumulative probabilities
and other plots that allow to visualize the average/minimum/maximum number of tokens in places
during the simulations. 
The existing TAPAAL GUI~\cite{DJJJMS:TACAS:12} is extended to allow for the editing of stochastic TAPNs as well as queries,
it communicates with the SMC engine and visualizes the results of SMC verification.
A screenshot using our running producer/consumer example with the window showing the result
with the statistics as well as the cumulative probability is shown in Figure~\ref{fig:tapaal-smc}.

\begin{figure}[t]
    \centering
    \scalebox{0.23}{\includegraphics[]{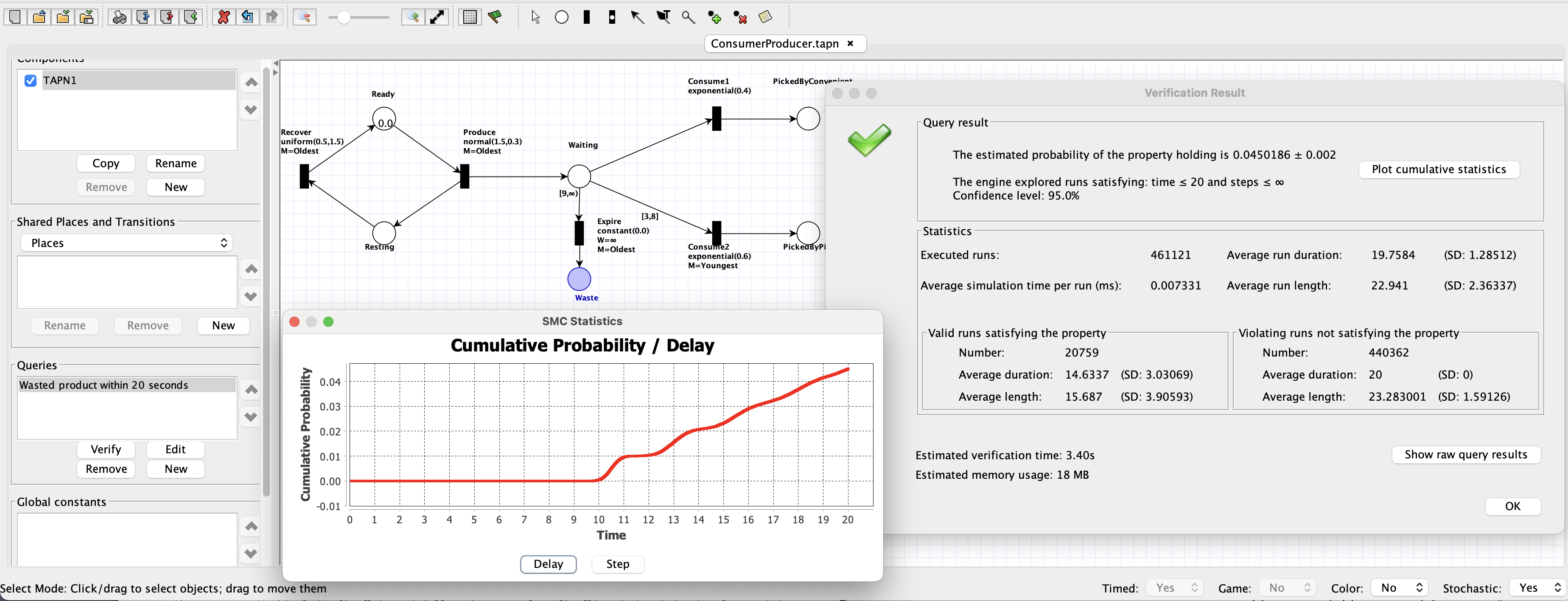}}
    \caption{TAPAAL's GUI with a running example and its SMC verification output}
    \label{fig:tapaal-smc}
\end{figure}

We currently support the uniform and geometric discrete distributions and
dirac delta, uniform, exponential, normal, gamma, triangular and log normal continuous
distributions.
In order to allow the user to sample from an arbitrary distribution, we also
introduce custom distributions where the user provides a list of concrete values that will be sampled on demand during the run generation. 
For the given set of distribution parameters, the GUI also visualizes the
density function.
We support three types of verification: quantitative probability estimation, using a Monte-Carlo algorithm; qualitative probability test, using a sequential probability ratio test; and a mode to generate and simulate traces, which can be either any traces, traces satisfying a property, or traces violating a property. 
The SMC engine has an option to use multi-threading and, if enabled, it executes one run generator per available core and achieves almost a linear speedup with increasing number of CPU cores. Moreover,
we support the estimation of the total running time of the SMC algorithm from the given confidence and precision parameters, or alternatively suggest a precision to match a given running time and confidence level.


We shall now present four case studies demonstrating the capabilities of our tool.
The most recent release of TAPAAL SMC can be downloaded
from \url{https://www.tapaal.net/download/}. The models used in the case studies can be found in other downloads at the bottom of the page.




\subsection{Fireflies}

Firefly is a beetle that is capable of emitting light flashes in periodic intervals. 
In nature, some firefly species like Photinus carolinus~\cite{wiki:Photinus_carolinus} show a remarkable synchronization capabilities so that a population of
fireflies starts, after a while, to flash in synchrony. All this without any leader and in a completely distributed
manner. This behaviour attracted the attention among scientists, trying to explain and model such behaviour
(see e.g.~\cite{fireflies-book,fireflies}). In a simplified version, each firefly requires some amount of time 
in order to charge and be able to flash. The discharge (flash) can then occur after some random delay, unless the firefly
notices a flash from other firefly, in which case it will join the flash (assuming that it is also charged).

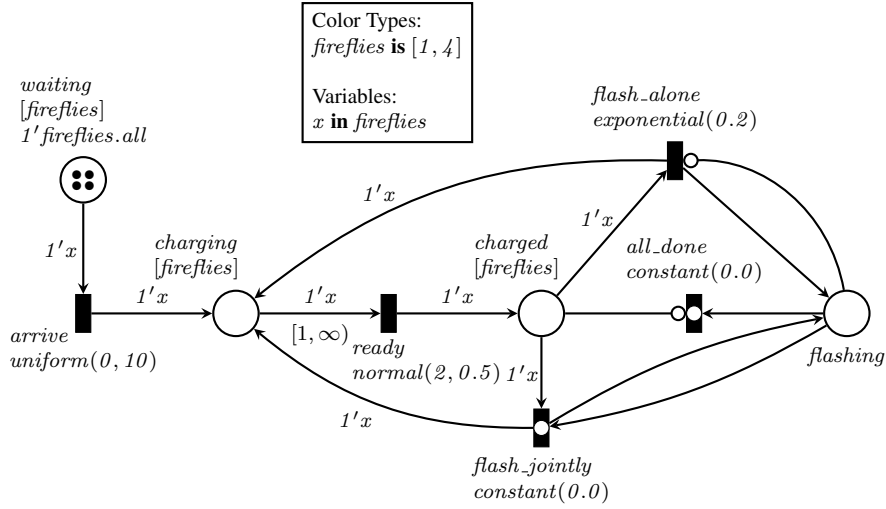
\begin{figure}[t]
    \centering
\begin{tikzpicture}[font=\scriptsize, xscale=0.36, yscale=0.36, x=1.33pt, y=1.33pt]
\tikzstyle{arc}=[->,>=stealth,thick]
\tikzstyle{transportArc}=[->,>=diamond,thick]
\tikzstyle{inhibArc}=[->,>=o,thick]
\tikzstyle{every place}=[minimum size=6mm,thick]
\tikzstyle{every transition} = [fill=black,minimum width=2mm,minimum height=5mm]
\tikzstyle{every token}=[fill=white,text=black]
\tikzstyle{sharedplace}=[place,minimum size=7.5mm,dashed,thin]
\tikzstyle{sharedtransition}=[transition, fill opacity=0, minimum width=3.5mm, minimum height=6.5mm,dashed]
\tikzstyle{urgenttransition}=[place,fill=white,minimum size=2.0mm,thin]
\tikzstyle{uncontrollabletransition}=[transition,fill=white,draw=black,very thick]
\tikzstyle{globalBox} = [draw,thick,align=left]
\node[place, label={[align=left,label distance=0cm]90:$\mathit{waiting}$\\$\mathit{[fireflies]}$ \\$\mathit{1'fireflies.all}$ }] at (270,-60) (waiting) {};
\node at (265.0,-56.0){$\bullet$};
\node at (265.0,-65.0){$\bullet$};
\node at (275.0,-56.0){$\bullet$};
\node at (275.0,-65.0){$\bullet$};

\node[place, label={[align=left,label distance=0cm,xshift=2mm]93:$\mathit{charging}$\\$\mathit{[fireflies]}$ }] at (390,-165) (charging) {};
\node[place, label={[align=left,label distance=0cm,xshift=4mm]94:$\mathit{charged}$\\$\mathit{[fireflies]}$ }] at (630,-165) (charged) {};
\node[place, label={[align=left,label distance=0cm]270:$\mathit{flashing}$ }] at (870,-165) (flashing) {};
\node[transition, label={[align=left,label distance=0cm, yshift=2mm]270:$\mathit{arrive}$\\$\mathit{uniform(0,10)}$ }] at (270,-165) (arrive) {};
\node[transition, label={[align=left,label distance=0cm,yshift=0.5mm,xshift=5mm]270:$\mathit{ready}$\\$\mathit{normal(2,0.5)}$ }] at (510,-165) (ready) {};
\node[transition, label={[align=left,label distance=0cm]90:$\mathit{flash\_alone}$\\$\mathit{exponential(0.2)}$ }] at (735,-45) (flash_alone) {};
\node[transition, label={[align=left,label distance=0cm]270:$\mathit{flash\_jointly}$\\$\mathit{constant(0.0)}$ }] at (630,-255) (flash_jointly) {};
\node[urgenttransition] at (flash_jointly.center) { };
\node[transition, label={[align=left,label distance=0cm]90:$\mathit{all\_done}$\\$\mathit{constant(0.0)}$ }] at (750,-165) (all_done) {};
\node[urgenttransition] at (all_done.center) { };
\draw[arc,pos=0.5] (waiting) to node[bend right=0,align=right,xshift=-3mm] {$\mathit{}$ $\mathit{1'x}$ } (arrive);
\draw[arc,pos=0.5] (arrive) to node[bend right=0,auto,align=left] {\\$\mathit{1'x}$ } (charging);
\draw[arc,pos=0.5] (charging) to node[bend right=0,auto,align=left] {$\mathit{}$ $\mathit{1'x}$ }
node[bend right=0,below,align=left] {$\mathit{}$ $[1,\infty)$ }(ready);
\draw[arc,pos=0.5,xshift=0mm] (ready) to node[bend right=0,auto,align=left] {\\$\mathit{1'x}$ } (charged);
\draw[arc,pos=0.5] (charged) to node[bend right=0,align=left,xshift=-2mm,yshift=2mm] {$\mathit{}$ $\mathit{1'x}$ } (flash_alone);
\draw[arc,pos=0.5] (charged) to node[bend right=0,auto,align=right,xshift=-7mm] {$\mathit{}$ $\mathit{1'x}$ } (flash_jointly);
\draw[arc,pos=0.5] (flash_alone) to node[bend right=0,auto,align=left] { } (flashing);
\draw[arc,pos=0.5,bend right=20] (flash_alone)  
to node[bend right=20,align=left,xshift=-9mm] {\\$\mathit{1'x}$ } (charging);
\draw[arc,pos=0.5, bend left=10] (flashing)  
to node[bend right=0,auto,align=left] {} (flash_jointly);
\draw[arc,pos=0.5, bend left=10] (flash_jointly) to node[bend right=0,auto,align=left] {} (flashing);
\draw[arc,pos=0.5,bend left=20] (flash_jointly)  
to node[bend right=0,auto,align=left,yshift=2mm] {\\$\mathit{1'x}$ } (charging);
\draw[arc,pos=0.5] (flashing) to node[bend right=0,auto,align=left] {} (all_done);
\draw[inhibArc,pos=0.5] (charged) to node[bend right=0,auto,align=left] {} (all_done);
\draw[inhibArc,pos=0.5, bend right=41] (flashing) 
to node[bend right=0,auto,align=left] {} (flash_alone);
\node [globalBox,xshift=40mm,yshift=-10mm] (globalBox) at (current bounding box.north west) [anchor=south west] {Color Types:\\$\mathit{fireflies}$ \textbf{is} $\mathit{[1,4]}$\\\\Variables:\\ $\mathit{x \textbf{ in } fireflies}$};
\end{tikzpicture}
    \caption{Colored stochastic Petri net model of fireflies}
    \label{fig:fireflies}
\end{figure}

In Figure~\ref{fig:fireflies} we present a stochastic timed-arc Petri net model of the fireflies flashing behaviour.
We use here the well-known colored extension of Petri nets where tokens can carry additional information
(in our case the id of the firefly). Initially, there are four fireflies (tokens) in the place \emph{waiting}, having ids one to four and initially of age $0$. With uniform distribution between 0 to 10 seconds, the fireflies independently arrive to the place \emph{charging}. Here each firefly must be charged before the transition \emph{ready} moves it to the charged location. The charging takes at least one second due to the
time interval $[1,\infty)$ on the arc. After the transition \emph{ready} is enabled, an additional delay is sampled from a normal distribution with mean value of 2 seconds and a standard deviation of 0.5.
Once a firefly arrives to the place \emph{charged}, there are two situations. Either the place \emph{flashing} has no token, meaning that there is currently no one emitting a flash. In this case the firefly waits an amount of time sampled from the exponential distribution of rate 0.2 (corresponding to the mean value of 5 seconds) and once it is scheduled to fire the transition \emph{flash\_alone}, it returns to the charging location. As a side-effect, a token (uncolored) is placed to the place \emph{flashing}, which disables other fireflies from flashing alone. A token in the place \emph{flashing} indicates that other fireflies in the place \emph{charged} should join the flash immediately. This is done by firing the transition \emph{flash\_jointly} which samples from the constant distribution which always returns the value 0. In other words, the transition \emph{firing} is scheduled immediately and the transition hence becomes urgent (indicated by the circle in the middle of the transition). As soon as there are no further fireflies in the place \emph{charged}, the urgent transition \emph{all\_done} becomes enabled and removes the token from the place \emph{flashing}.

\begin{figure}[t]
    \centering
    \scalebox{0.27}{\includegraphics[]{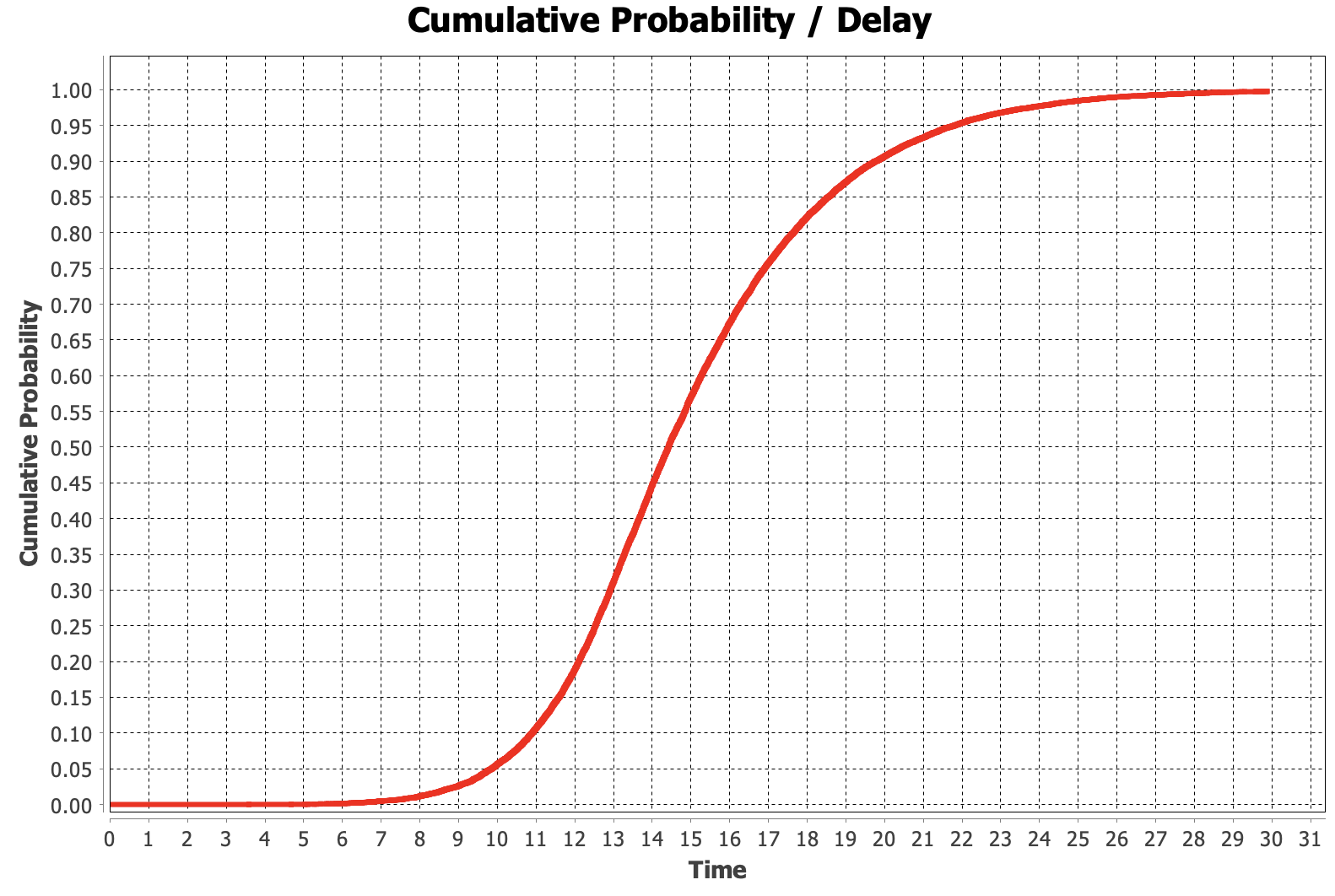}}~~~
    \scalebox{0.27}{\includegraphics[]{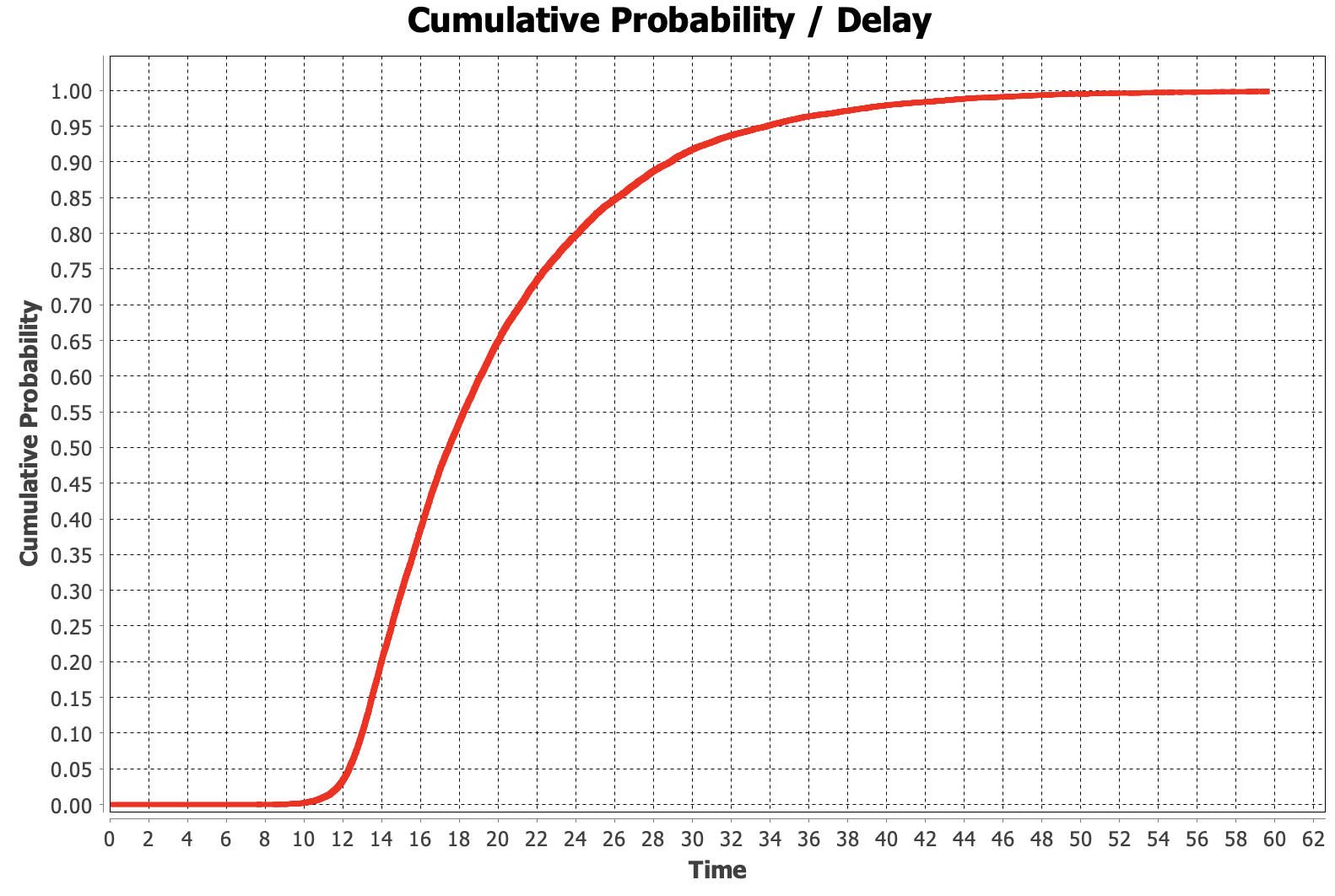}}
    \caption{Cumulative synchronization probability for 4 (left) and 10 (right) fireflies}
    \label{fig:fireflies-plot}
\end{figure}

In our tool, a colored Petri net is unfolded~\cite{CGSS:PERFORMANCE:20,BJPST:FI:23} to the standard timed-arc Petri net without colors by creating a~copy of each place for each color and a~copy of each transition for each binding of the variable $x$ to different fireflies ids. The probability distributions of each transition are then simply 
overtaken by the unfolded copies of the transition. 

We can now ask about the probability that all fireflies synchronize within 30 seconds. This can be formulated by the query 
\begin{center}
{\tt F (charging = 1 and flashing = 1 and waiting = 0)} 
\end{center} 
requiring that we reach a situation where there are no fireflies in the place \emph{waiting}, exactly one firefly in the place \emph{charging} and a token in \emph{flashing}, indicating that a single firefly initiated a flash and all other fireflies (that are in the place \emph{charged}) are joining in a synchronous flash.
Our tool executes 18456 runs to estimate that the probability of synchronization within 30 seconds
is $0.997$ with $95$\% confidence and precision $0.01$. If we increase the number of fireflies to 10, the probability
of synchronous flashing within 30 seconds
drops to $0.917$ and we need at least 60 seconds to achieve the synchronization probability $0.998$.
Cumulative probability plots for 4 and 10 fireflies, produced by our tool, are depicted in Figure~\ref{fig:fireflies-plot}.

\subsection{Frequency Spectrum As signment in Elastic Optical Networks}

Elastic all-optical networks~\cite{all-opt-net} allow for fine-grained resource allocation technologies 
in order to schedule network traffic demands on different light frequencies inside a single
optical fiber. When a new demand arrives, the spectrum allocation problem is to find
for each optical fiber a sequence of available frequency slots that will carry the demand.
The frequency slots must be contiguous~\cite{rsa-problem} (follow each other) and the number of required slots depends on the modulation
scheme~\cite{hideki-modulation}. After some time, a demand can be released, making the frequency slots allocated
for the demand available again. This allocation/deallocation process can create a fragmentation in the allocated
frequency slots, possibly resulting in a situation when a newly arrived demand is blocked
(there are not enough available consecutive frequency slots to accommodate the demand).

\begin{figure}[t]
    \centering
   \scalebox{0.75}{ \begin{tikzpicture}[font=\scriptsize, xscale=0.45, yscale=0.45, x=1.33pt, y=1.33pt]
\tikzstyle{globalBox} = [draw,thick,align=left]
\node [globalBox,xshift=0mm] (globalBox) at (0,0)  {
Color Types:\\$\mathit{Modulation}$ \textbf{is} $[\textit{16QAM, 8QAM, QPSK}]$\\$\mathit{Slots}$ \textbf{is} $\mathit{[0,20]}$\\$\mathit{SpectrumSize}$ \textbf{is} $\mathit{[3,6]}$\\$\mathit{Occupied}$ \textbf{is} [$\mathit{Slots, SpectrumSize}$]\\\\Variables:\\ $\mathit{mod \textbf{ in } Modulation}$\\$\mathit{s \textbf{ in } Slots}$};
\end{tikzpicture} }\\
    \scalebox{0.72}{

\begin{tikzpicture}[font=\scriptsize, xscale=0.45, yscale=0.45, x=1.33pt, y=1.33pt]
\tikzstyle{arc}=[->,>=stealth,thick]
\tikzstyle{transportArc}=[->,>=diamond,thick]
\tikzstyle{inhibArc}=[->,>=o,thick]
\tikzstyle{every place}=[minimum size=6mm,thick]
\tikzstyle{every transition} = [fill=black,minimum width=2mm,minimum height=5mm]
\tikzstyle{every token}=[fill=white,text=black]
\tikzstyle{sharedplace}=[place,minimum size=7.5mm,dashed,thin]
\tikzstyle{sharedtransition}=[transition, fill opacity=0, minimum width=3.5mm, minimum height=6.5mm,dashed]
\tikzstyle{urgenttransition}=[place,fill=white,minimum size=2.0mm,thin]
\tikzstyle{uncontrollabletransition}=[transition,fill=white,draw=black,very thick]
\tikzstyle{globalBox} = [draw,thick,align=left]
\node[place, label={[align=left,label distance=0cm]90:$\mathit{demands}$\\$\mathit{[Modulation]}$ }] at (405,-75) (demands) {};
\node[sharedplace ] at (demands.center) { };
\node[place, label={[align=left,label distance=0cm]180:$\mathit{free}$\\$\mathit{[Slots]}$ \\$\mathit{1'Slots.all}$ }] at (320,-285) (free) {};
\node at (320.0,-285.0){$\mathrm{\#21}$};
\node[sharedplace ] at (free.center) { };
\node[place, label={[align=left,label distance=0cm]180:$\mathit{occupied}$\\$\mathit{[Occupied]}$ }] at (490,-285) (occupied) {};
\node[sharedplace ] at (occupied.center) { };
\node[place, label={[align=left,label distance=0cm,xshift=1mm
]270:$\mathit{blocked}$ }] at (555,-180) (blocked) {};
\node[sharedplace ] at (blocked.center) { };
\node[transition,rotate=-180, label={[align=left,label distance=0cm]270:$\mathit{arrive}$\\$\mathit{exponential(0.1)}$ }] at (255,-75) (arrive) {};
\node[transition, label={[align=left,label distance=0cm]90:$\mathit{block}$\\$\mathit{constant(0)}$ \\$\mathit{W=0}$ }] at (555,-75) (block) {};
\node[urgenttransition] at (555,-75) (blockurgent) {};
\node[transition, label={[align=left,label distance=0cm,xshift=3mm]270:$\mathit{release}$\\$\mathit{normal(3,1)}$ }] at (405,-390) (release) {};
\node[transition, label={[align=left,label distance=0cm,yshift=3mm]180:$\mathit{allocate}$\\$\mathit{constant(0)}$ \\$\mathit{s}$ $\mathit{\leq}$ $\mathit{18}$ }] at (405,-180) (allocate) {};
\node[urgenttransition] at (405,-180) (allocateurgent) {};
\draw[arc,pos=0.5] (demands) to node[bend right=0,below,align=left] {$\mathit{}$ $\mathit{1'mod}$ } (block);
\draw[arc,pos=0.5] (occupied) to node[bend right=0,auto,align=left,xshift=-2mm] {$\mathit{}$ $\mathit{1'(s,}$ $\mathit{3)}$ } node[bend right=0,above,align=left,xshift=-2mm] {$[1,\infty)$ }  (release);
\draw[arc,pos=0.5] (demands) to node[bend right=0,auto,align=left] {$\mathit{}$ $\mathit{1'\textit{16QAM}}$ } (allocate);
\draw[arc,pos=0.5] (free) to node[bend right=0,auto,align=left,xshift=2mm] {$\mathit{}$ $\mathit{1's}$ $\mathit{+}$ $\mathit{1'(s\!+\!1)}$ $\mathit{+}$ $\mathit{1'(s\!+\!2)}$ } (allocate);
\draw[arc,pos=0.5] (arrive) to node[bend right=0,below,align=left] {\\$\mathit{1'mod}$ } (demands);
\draw[arc,pos=0.5] (block) to node[bend right=0,auto,align=left] { } (blocked);
\draw[arc,pos=0.5] (release) to node[bend right=0,auto,align=left,xshift=2mm] {\\$\mathit{1's}$ $\mathit{+}$ $\mathit{1'(s\!+\!1)}$ $\mathit{+}$ $\mathit{1'(s\!+\!2)}$ } (free);
\draw[arc,pos=0.5] (allocate) to node[bend right=0,auto,align=left,xshift=-2mm] {\\$\mathit{1'(s,}$ $\mathit{3)}$ } (occupied);
\end{tikzpicture}} \hspace{-2mm}
    \scalebox{0.72}{
\begin{tikzpicture}[font=\scriptsize, xscale=0.45, yscale=0.45, x=1.33pt, y=1.33pt]
\tikzstyle{arc}=[->,>=stealth,thick]
\tikzstyle{transportArc}=[->,>=diamond,thick]
\tikzstyle{inhibArc}=[->,>=o,thick]
\tikzstyle{every place}=[minimum size=6mm,thick]
\tikzstyle{every transition} = [fill=black,minimum width=2mm,minimum height=5mm]
\tikzstyle{every token}=[fill=white,text=black]
\tikzstyle{sharedplace}=[place,minimum size=7.5mm,dashed,thin]
\tikzstyle{sharedtransition}=[transition, fill opacity=0, minimum width=3.5mm, minimum height=6.5mm,dashed]
\tikzstyle{urgenttransition}=[place,fill=white,minimum size=2.0mm,thin]
\tikzstyle{uncontrollabletransition}=[transition,fill=white,draw=black,very thick]
\tikzstyle{globalBox} = [draw,thick,align=left]
\node[place, label={[align=left,label distance=0cm]90:$\mathit{demands}$\\$\mathit{[Modulation]}$ }] at (375,-45) (demands) {};
\node[sharedplace ] at (demands.center) { };
\node[place, label={[align=left,label distance=0cm]180:$\mathit{free}$\\$\mathit{[Slots]}$ \\$\mathit{1'Slots.all}$ }] at (245,-270) (free) {};
\node at (245.0,-270.0){$\mathrm{\#21}$};
\node[sharedplace ] at (free.center) { };
\node[place, label={[align=left,label distance=0cm]180:$\mathit{occupied}$\\$\mathit{[Occupied]}$ }] at (415,-270) (occupied) {};
\node[sharedplace ] at (occupied.center) { };
\node[place, label={[align=left,label distance=0cm]90:$\mathit{blocked}$ }] at (615,-45) (blocked) {};
\node[sharedplace ] at (blocked.center) { };
\node[place, label={[align=left,label distance=0cm,yshift=-2mm]100:$\mathit{counter}$\\$\mathit{[Slots]}$ }] at (495,-165) (counter) {};
\node[transition,rotate=-180, label={[align=left,label distance=0cm]270:$\mathit{arrive}$\\$\mathit{exponential(0.1)}$ }] at (270,-45) (arrive) {};
\node[transition, label={[align=left,label distance=0cm,xshift=3mm]270:$\mathit{release}$\\$\mathit{normal(3,1)}$ }] at (330,-375) (release) {};
\node[transition, label={[align=left,label distance=0cm]90:$\mathit{allocate}$\\$\mathit{constant(0.0)}$ \\$\mathit{W=\infty}$ \\$\mathit{s}$ $\mathit{\leq}$ $\mathit{18}$ }] at (330,-165) (allocate) {};
\node[urgenttransition] at (330,-165) (allocateurgetn) {};
\node[transition, label={[align=left,label distance=0cm,xshift=3mm]90:$\mathit{init}$\\$\mathit{constant(0)}$ }] at (495,-45) (init) {};
\node[urgenttransition] at (495,-45) (initurgent) {};
\node[transition, label={[align=left,label distance=0cm,xshift=2mm
]270:$\mathit{increment}$\\$\mathit{constant(0)}$ \\$\mathit{W=0}$ }] at (495,-270) (increment) {};
\node[urgenttransition] at (495,-270) (incrementurgent) {};
\node[transition, label={[align=left,label distance=0cm]270:$\mathit{block}$\\$\mathit{constant(0.0)}$ \\$\mathit{s}$ $\mathit{>=}$ $\mathit{18}$ }] at (615,-165) (block) {};
\node[urgenttransition] at (615,-165) (blockurgent) {};
\draw[arc,pos=0.5] (occupied) to node[bend right=0,auto,align=left,xshift=-2mm] {$\mathit{}$ $\mathit{1'(s,}$ $\mathit{3)}$ } node[bend right=0,above,align=left,xshift=-2mm] {$[1,\infty)$ } (release);
\draw[arc,pos=0.3] (free) to node[bend right=0,auto,align=left,xshift=2mm] {$\mathit{}$ $\mathit{1's}$ $\mathit{+}$ $\mathit{1'(s\!+\!1)}$ $\mathit{+}$ $\mathit{1'(s\!+\!2)}$ } (allocate);
\draw[arc,pos=0.5] (demands) to node[bend right=0,below,align=left] {$\mathit{}$ $\mathit{1'\textit{16QAM}}$ } (init);
\draw[arc,pos=0.5] (counter) to node[bend right=0,above,align=left] {$\mathit{}$ $\mathit{1's}$ } (allocate);
\draw[arc,pos=0.5, bend right=20] (counter) to  node[bend right=0,left,align=left,xshift=1mm] {$\mathit{}$ $\mathit{1's}$ } (increment);
\draw[arc,pos=0.5] (counter) to node[bend right=0,auto,align=left] {$\mathit{}$ $\mathit{1's}$ } (block);
\draw[arc,pos=0.5] (arrive) to node[bend right=0,below,align=left] {\\$\mathit{1'mod}$ } (demands);
\draw[arc,pos=0.4] (release) to node[bend right=0,auto,align=left,xshift=2mm] {\\$\mathit{1's}$ $\mathit{+}$ $\mathit{1'(s\!+\!1)}$ $\mathit{+}$ $\mathit{1'(s\!+\!2)}$ } (free);
\draw[arc,pos=0.5] (init) to node[bend right=0,auto,align=left] {\\$\mathit{1'0}$ } (counter);
\draw[arc,pos=0.5] (allocate) to node[bend right=0,auto,align=left,xshift=-2mm] {\\$\mathit{1'(s,}$ $\mathit{3)}$ } (occupied);
\draw[arc,pos=0.5, bend right=20,xshift=-1mm] (increment) to 
node[bend right=0,right,align=left] {\\$\mathit{1'(s\!+\!1)}$ } (counter);
\draw[arc,pos=0.5] (block) to node[bend right=0,auto,align=left] { } (blocked);
\end{tikzpicture}}
    \caption{Random fit (left) and first fit (right) spectrum allocation for 16QAM}
    \label{fig:spectrum}
\end{figure}
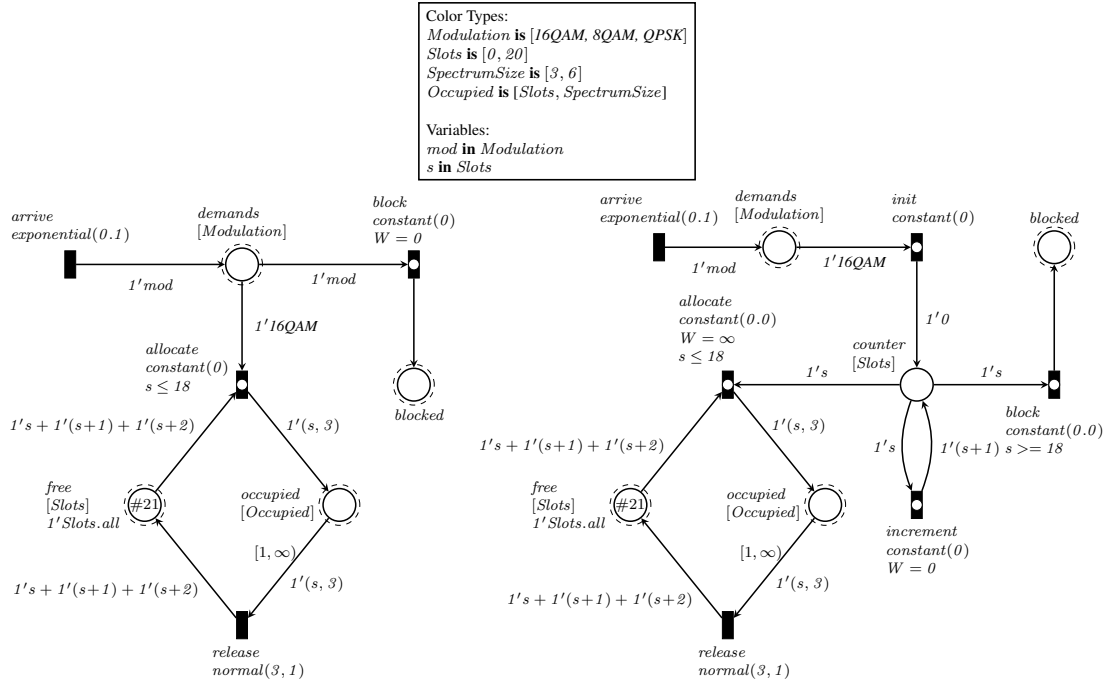

In our stochastic timed-arc Petri net model depicted in Figure~\ref{fig:spectrum}, we study
two spectrum allocation strategies: a \emph{random fit} that randomly chooses from the available slots,
trying to uniformly distribute the demands across the whole spectrum, and a strategy called \emph{first fit}~\cite{fist-random-fit} that always
uses the lowest possible available frequency slot.
We consider three common modulation schemes 16QAM, 8QAM and QPSK that require 3, 4 and 6 frequency slots~\cite{hideki-modulation}, respectively. In Figure~\ref{fig:spectrum} we present the Petri net model for 16QAM only but the
other two modulations are modelled anologously, just requiring a higher number of slots.
We use again the colored extension in order to allow for a more compact model description.
In our model, we use the variable $\mathit{mod}$ that ranges over the color type \textit{Modulation},
the variable $s$ ranging over the 21 frequency slots declared in the color type \textit{Slots}
and we also use the color type \textit{Occupied} which is a Cartesian product of the
\textit{Slots} color type and the available spectrum size with values 3 to 6.
The dashed circles around places denote the so-called shared places that appear also in the nets for the other modulations (not shown in our Petri net model). The shared place called \emph{free} contains
tokens with colors (frequency slots) that are currently available for allocation of a new demand.
Initially, there are 21 tokens in the place \emph{free} with color values from 0 to 20.
On the other hand, the shared place \emph{occupied} contains tokens with colors being pairs of integers
$(\mathit{frequency},\mathit{size})$ such that e.g. the value $10,4$ means that
the frequency slots 10, 11, 12 and 13 are occupied by a demand that requires $4$ consecutive
frequency slots.

The random fit slot allocation (Figure~\ref{fig:spectrum} left) gathers the
arrived demands (assuming an exponential distribution of arrivals with rate $0.1$) in the place
called \emph{demands}. Immediately after a demand arrival, the urgent transition \emph{block} becomes enabled
but it has weight $0$, meaning that if the transition \emph{allocate} is also enabled, it will have a priority.
The transition \emph{allocate} binds $s$ to some initial frequency slot (no more than 18 as the demand requires
three slots) and checks whether the frequency slots $s$, $s+1$ and $s+2$ are available (i.e. there are three tokens with these colors in the the place \emph{free}). If this is the case, these three tokens are removed
from the place \emph{free} and a new token with the color $(s,3)$ is added to \emph{occupied}. After one second holding time, enforced by the interval $[1,\infty)$, 
the transition \emph{release} becomes enabled and samples its firing date from the normal distribution with mean $3$ and standard deviation $1$ and releases these three slots once it fires. In case the transition \emph{allocate} is enabled
for several bindings of the variable $s$ to different frequency slots, the used binding is
selected using a uniform distribution.

In Figure~\ref{fig:spectrum} (right)  we describe the first fit spectrum allocation algorithm
that always uses the lowest available frequency slots. This is achieved by placing a token
with the color $0$ (the lowest frequency slot) to the place \emph{counter} immediately after any demand arrival.
By firing the urgent transition
\emph{increment}, we then repeatedly increase the color value of the token by 1 until we find an available slot
(using the same modelling approach as in the random assignment) or until we reach the frequency
slot $18$ where the transition \emph{block} becomes enabled. The weight of the transition \emph{allocate}
is infinity, meaning that it has a priority over the transition \emph{block} but in case the transition
is still not enabled, the transition \emph{block} will fire as the transition increment has weight 0.

We can now ask our tool to compare the blocking probability of the two allocation algorithms
by issuing the SMC query {\tt F blocked >= 1}, checking whether at least one demand is blocked
within the first 30 seconds. 
With 95\% confidence and precision $0.01$,
our tool executes 18488 runs and estimates the blocking probability to $0.075$ for the
first fit slot allocation strategy and $0.167$ for the random fit strategy, showing an improved performance of the former strategy.
This behaviour is also conformed by experiments in computer networking literature~\cite{fist-random-fit}.

\subsection{Manufacturing Process of Wooden Windows}

\begin{figure}[t]
    \centering
    \scalebox{0.49}{
    \input{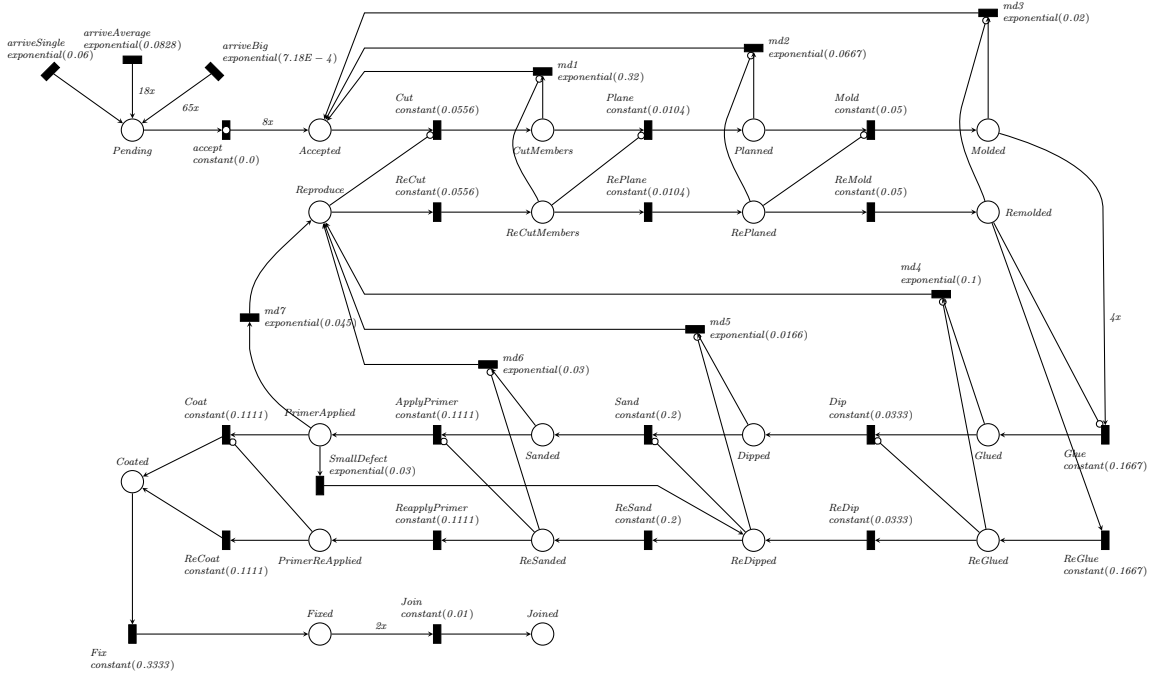}
    }
    \caption{Manufacturing process of wooden windows~\cite{woden:windows}}
    \label{fig:wood}
\end{figure}

The workflow of a Hungarian wooden window manufacturing company Holz-Team Ltd details a systematic production process, from preparing individual frame components to assembling and finishing complete windows. Orders vary in size, and the process includes measures to address material defects, with early-stage issues being easier to resolve than those found later. The Petri net model is given in Figure~\ref{fig:wood} and contains timing
probabilities derived from real measurements performed by the company as described in~\cite{woden:windows}.
The workflow uses two types of transitions: exponentially delayed transitions for unpredictable events like order arrivals and material defects, and deterministic transitions for fixed-time production phases such as cutting, molding, and coating. For a detailed description of the process consult~\cite{woden:windows}.

The production process faced a bottleneck where coated windows that required a hardware fix accumulated over time in the place called \emph{Coated}. Indeed, using TAPAAL SMC engine we estimated
that with probability $0.61 \pm 0.01$ and $95$\% confidence, the weekly production process generates
more than 20 unprocessed coated windows. A solution to the bottleneck was suggested in~\cite{woden:windows} and required the hiring of an additional worker to reduce
the hardware fixing operation from 20 minutes to 13 minutes. After this adjustment, the probability
of storing 4 or more unprocessed windows dropped to 0.011. In Figure~\ref{fig:windows} (left) we depict, based on 50 simulations of the workflow for a period of one year (2008 hours), the number of accepted
orders and the average number of unprocessed coated windows for the
situation before the workflow adjustment. After the adjustment, the number of unprocessed coated
windows stays on average slightly above 1 with a maximum between 3 to 5 windows, as shown
in Figure~\ref{fig:windows} (right). This allowed the company to produce on average about 260 additional windows per year, increasing the productivity by almost 10\%.

\begin{figure}[t]
    \centering
    \scalebox{0.27}{\includegraphics[]{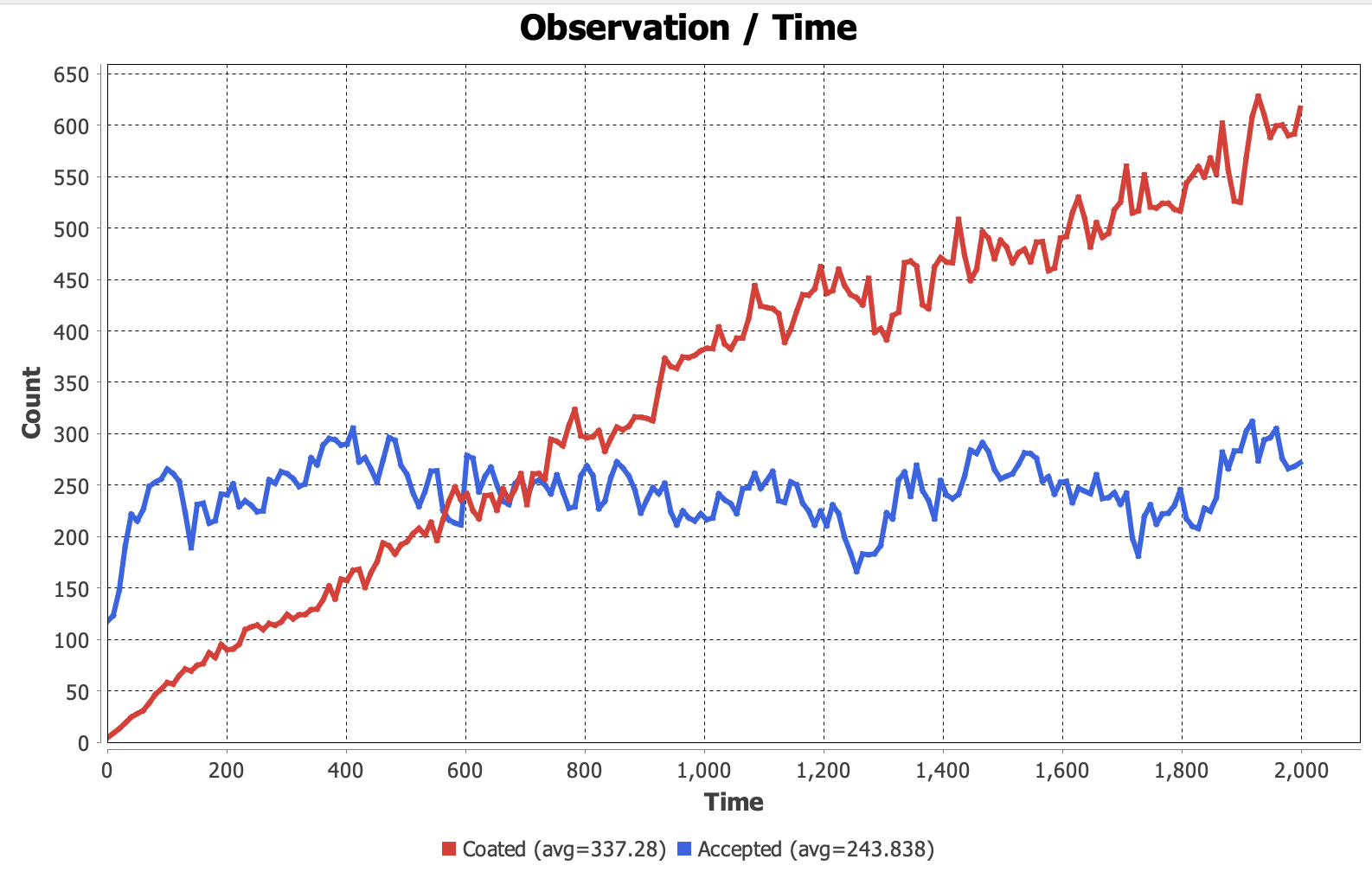}}~~~
    \scalebox{0.27}{\includegraphics[]{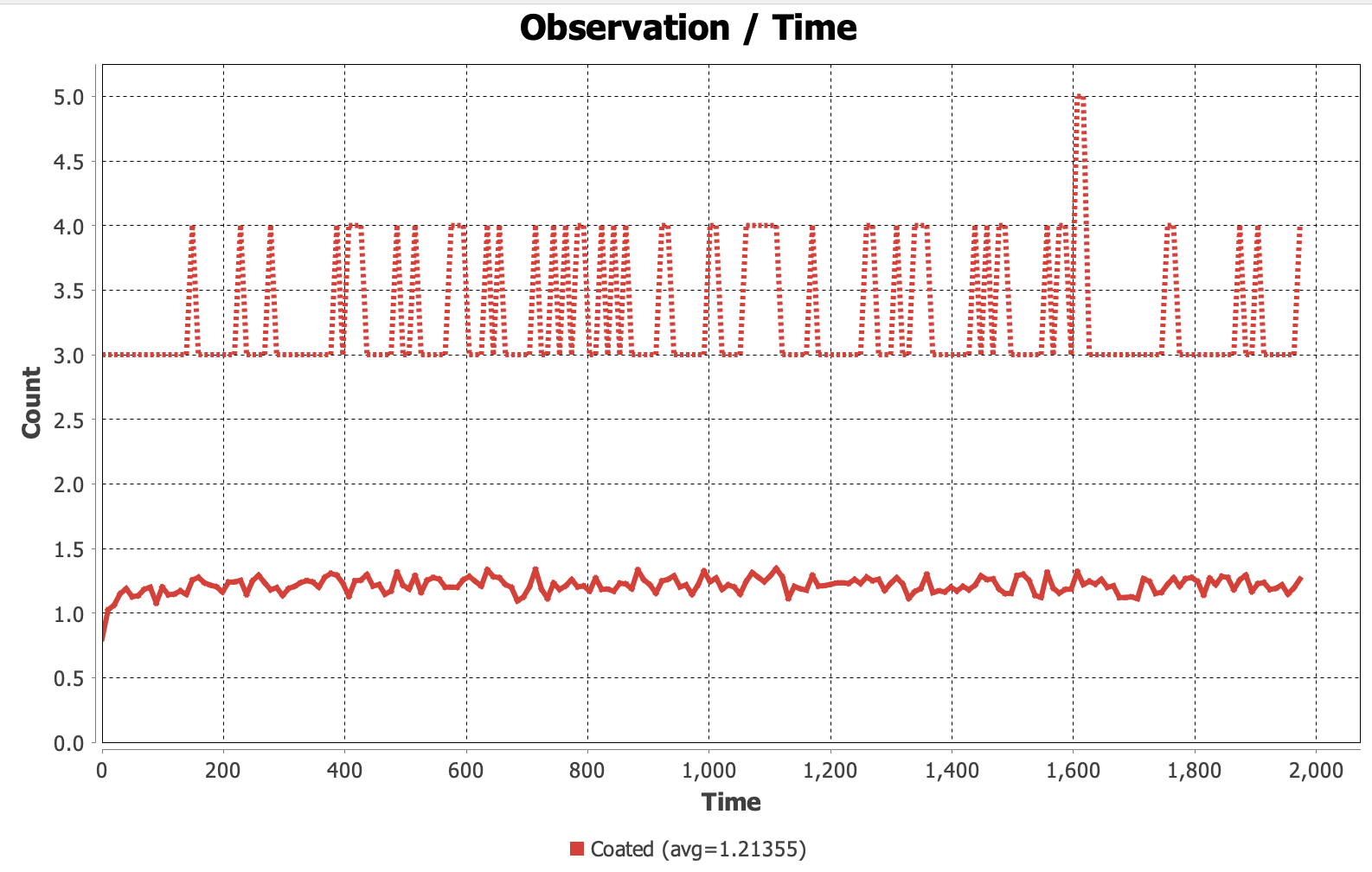}}
    \caption{Average number of orders and unprocessed coated windows over a year}
    \label{fig:windows}
\end{figure}

\subsection{Quantum Communication---Beyond Super Dense Coding}

Super dense coding~\cite{Bennett1992} is a quantum communication procedure that allows a sender to transmit to a receiver two classical bits of information by sending only a single qubit, provided that the sender and receiver pre-share a maximally entangled pair of qubits (called an EPR/Bell pair~\cite{NielsenChuang2010} such as 
$|\phi^+\rangle = \frac{1}{\sqrt{2}}(|00\rangle + |11\rangle)$).

A recently proposed quantum protocol \emph{Beyond Super Dense Coding} (BSDC)~\cite{jensen2025quantum}
builds up on this idea and operates in a synchronized slotted-time manner over a~quantum channel and allows the protocol participants to communicate up to $1.5$ classical bits per
time slot while transmitting at most one qubit in each slot.
In BSDC, time is divided into discrete slots of a given frequency where the sender either stays silent (encoding the communication of a classical bit $0$)
or sends a qubit (encoding the classical bit $1$).
A communication round over a quantum channel proceeds as follows.
\begin{itemize}
\item     
The sender scans its incoming buffer of randomly generated information bits.
If the scanned bit is $0$, it remains silent in the given time slot.
Receiver interprets no activity on the quatum channel as receiving the bit $0$.
\item 
Upon the first arrival of the classical bit $1$, the sender on-the-fly generates an 
EPR pair of qubits, stores one qubit locally in quantum memory, and transmits the other maximally entangled qubit to the receiver. The receiver detects the arriving qubit (registering the arrival of the classical bit~$1$) and stores it locally without disturbing the entanglement.
\item
The sender then continues with being silent for every scanned bit $0$ in
its input buffer. The receiver registers this as receiving $0$ in each time slot
with no communication.
\item Once the second bit $1$ appears in the sender's buffer, the sender
applies Pauli operations~\cite{NielsenChuang2010} (X, Z, or identity) to its stored qubit, chosen according to the next two bits in its buffer. If the next bit is $1$, it applies the gate X and if the next-next bit is $1$, it applies also the gate Z; otherwise
it applies the identity operation. The sender then transmits this manipulated qubit to the receiver that interprets its arrival as the receipt of the classical bit $1$ and now, holding both qubits of the EPR pair, the receiver performs the CNOT controlled gate
on the both qubits, followed by the application of the gate H on the second received qubit. The two qubits that are now measured and the two resulting classical
bits are the two bits encoded by the sender. 
In this final slot, three information bits are thus transmitted in total
and the protocol can restart.
\end{itemize}

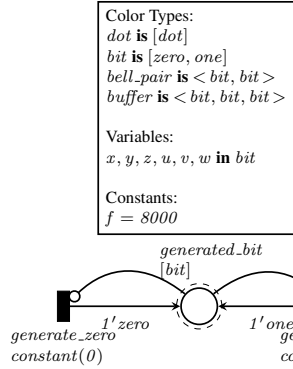
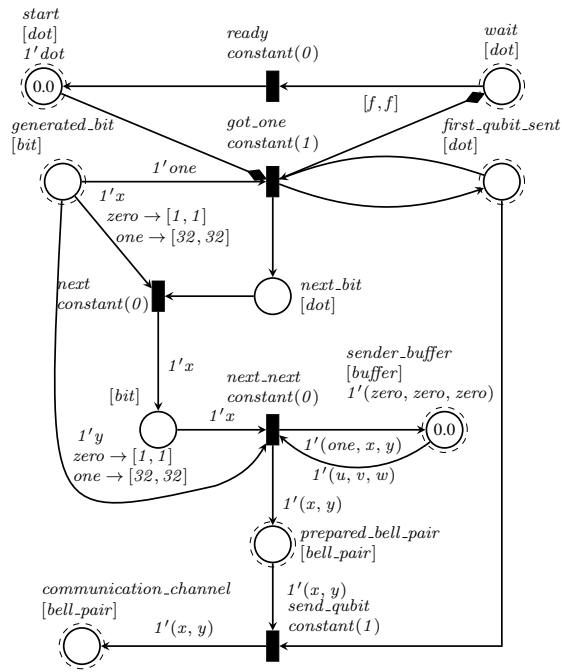
\begin{figure}
    \centering
    \begin{subfigure}[b]{0.33\textwidth}
        \centering
        \scalebox{0.8}{
\begin{tikzpicture}[font=\scriptsize, xscale=0.3, yscale=0.45, x=1.33pt, y=1.33pt]
\tikzstyle{arc}=[->,>=stealth,thick]
\tikzstyle{transportArc}=[->,>=diamond,thick]
\tikzstyle{inhibArc}=[->,>=o,thick]
\tikzstyle{every place}=[minimum size=6mm,thick]
\tikzstyle{every transition} = [fill=black,minimum width=2mm,minimum height=5mm]
\tikzstyle{every token}=[fill=white,text=black]
\tikzstyle{sharedplace}=[place,minimum size=7.5mm,dashed,thin]
\tikzstyle{sharedtransition}=[transition, fill opacity=0, minimum width=3.5mm, minimum height=6.5mm,dashed]
\tikzstyle{urgenttransition}=[place,fill=white,minimum size=2.0mm,thin]
\tikzstyle{uncontrollabletransition}=[transition,fill=white,draw=black,very thick]
\tikzstyle{globalBox} = [draw,thick,align=left]
\node[place, label={[align=left,xshift=2mm, label distance=0cm]90:$\mathit{generated\_bit}$\\$\mathit{[bit]}$ }] at (475,-135) (generated_bit) {};
\node[sharedplace ] at (generated_bit.center) { };
\node[transition,rotate=-180, label={[align=left,label distance=0cm]90:$\mathit{generate\_one}$\\$\mathit{constant(0)}$ }] at (630,-135) (generate_one) {};
\node[transition,rotate=-180, label={[align=left,label distance=0cm]90:$\mathit{generate\_zero}$\\$\mathit{constant(0)}$ }] at (315,-135) (generate_zero) {};
\draw[arc,pos=0.5] (generate_one) to node[bend right=0,auto,align=left,yshift=1mm] {\\$\mathit{1'one}$ } (generated_bit);
\draw[inhibArc,pos=0.5,bend left=30] (generated_bit)  to node[bend right=20,auto,align=left] {} (generate_one);
\draw[arc,pos=0.5] (generate_zero) to node[bend right=0,below,yshift=1mm, align=right] {\\$\mathit{1'zero}$ } (generated_bit);
\draw[inhibArc,pos=0.5, bend right=30] (generated_bit)  to node[bend right=0,auto,align=left] {} (generate_zero);
\node [globalBox] (globalBox) at (current bounding box.north) [anchor=south] {Color Types:\\$\mathit{dot}$ \textbf{is} $\mathit{[dot]}$\\ $\mathit{bit}$ \textbf{is} $\mathit{[zero, one]}$\\$\mathit{bell\_pair}$ \textbf{is} $\mathit{<\!bit, bit\!>}$\\$\mathit{buffer}$ \textbf{is} $\mathit{<\!bit, bit, bit\!>}$\\\\Variables:\\ $\mathit{x, y, z, u, v, w \textbf{ in } bit}$\\\\Constants:\\$\mathit{f = 8000}$};
\end{tikzpicture}
        }
        \caption{Classical bit generator}
        \label{fig:bit_generator}
     \end{subfigure}
     \begin{subfigure}[b]{0.59\textwidth}
        \centering
        \scalebox{0.8}{
\begin{tikzpicture}[font=\scriptsize, xscale=0.45, yscale=0.45, x=1.33pt, y=1.33pt]
\tikzstyle{arc}=[->,>=stealth,thick]
\tikzstyle{transportArc}=[->,>=diamond,thick]
\tikzstyle{inhibArc}=[->,>=o,thick]
\tikzstyle{every place}=[minimum size=6mm,thick]
\tikzstyle{every transition} = [fill=black,minimum width=2mm,minimum height=5mm]
\tikzstyle{every token}=[fill=white,text=black]
\tikzstyle{sharedplace}=[place,minimum size=7.5mm,dashed,thin]
\tikzstyle{sharedtransition}=[transition, fill opacity=0, minimum width=3.5mm, minimum height=6.5mm,dashed]
\tikzstyle{urgenttransition}=[place,fill=white,minimum size=2.0mm,thin]
\tikzstyle{uncontrollabletransition}=[transition,fill=white,draw=black,very thick]
\tikzstyle{globalBox} = [draw,thick,align=left]
\node[place, label={[align=left,label distance=0cm]180:$\mathit{generated\_bit}$\\$\mathit{[bit]}$ }] at (270,-75) (generated_bit) {};
\node[sharedplace ] at (generated_bit.center) { };
\node[place, label={[align=left,label distance=0cm]270:$\mathit{start}$\\$\mathit{[dot]}$ \\$\mathit{1'dot}$ }] at (270,-315) (start) {};
\node at (270.0,-315.0){0.0};
\node[sharedplace ] at (start.center) { };
\node[place, label={[align=left,label distance=0cm,xshift=2mm,yshift=-1mm]270:$\mathit{sender\_buffer}$\\$\mathit{[buffer]}$ \\$\mathit{1'(zero,}$ $\mathit{zero,}$ $\mathit{zero)}$ }] at (105,-165) (sender_buffer) {};
\node at (105.0,-165.0){0.0};
\node[sharedplace ] at (sender_buffer.center) { };
\node[place, label={[align=left,label distance=0cm]180:$\mathit{wait}$\\$\mathit{[dot]}$ }] at (105,-315) (wait) {};
\node[sharedplace ] at (wait.center) { };
\node[place, label={[align=left,label distance=0cm]270:$\mathit{communication\_channel}$\\$\mathit{[bell\_pair]}$ }] at (450,-340) (communication_channel) {};
\node[sharedplace ] at (communication_channel.center) { };
\node[place, label={[align=left,label distance=0cm]270:$\mathit{first\_qubit\_sent}$\\$\mathit{[dot]}$ }] at (520,-160) (first_qubit_sent) {};
\node[sharedplace ] at (first_qubit_sent.center) { };
\node[transition, label={[align=left,label distance=0cm]0:$\mathit{got\_zero}$\\$\mathit{constant(1)}$ }] at (270,-165) (got_zero) {};
\node[transition, label={[align=left,xshift=2mm,label distance=0cm]90:$\mathit{ready}$\\$\mathit{constant(0)}$ }] at (180,-315) (ready) {};
\node[transition, label={[align=left,label distance=0cm]20:$\mathit{got\_one}$\\$\mathit{normal(650,50)}$ }] at (450,-75) (got_one) {};
\draw[arc,pos=0.5] (generated_bit) to node[bend right=0,auto,align=left] {$\mathit{}$ $\mathit{1'zero}$ } (got_zero);
\draw[arc,pos=0.5, bend left=20] (sender_buffer)  to node[bend right=0,auto,align=left] {$\mathit{}$ $\mathit{1'(x,}$ $\mathit{y,}$ $\mathit{z)}$ } (got_zero);
\draw[arc,pos=0.5, bend left=20] (got_zero)  to node[bend right=0,auto,above, align=left,yshift=0mm] {\\$\mathit{1'(y,}$ $\mathit{z,}$ $\mathit{zero)}$ } (sender_buffer);
\draw[arc,pos=0.5,yshift=3mm] (wait) to node[bend right=0,below,align=left,yshift=1mm] {$\mathit{}$  \\$\mathit{[f,f]}$ } (ready);
\draw[arc,pos=0.5] (ready) to node[bend right=0,auto,align=left] {} (start);
\draw[transportArc,pos=0.5] (start) to node[bend right=0,auto,align=left] {$\mathit{}$  } (got_zero);
\draw[transportArc,pos=0.5] (got_zero)  .. controls(254.0,-189.0) and (266.0,-255.0) .. (237,-252) to [bend right=0] (105,-252) to node[bend right=0,auto,align=left] {$\mathit{}$  } (wait);
\draw[transportArc,pos=0.5] (start) to node[bend right=0,auto,align=left] {$\mathit{}$ $\mathit{}$} (got_one);
\draw[transportArc,pos=0.5] (got_one) to [bend right=0] (300,-380) to [bend right=0] (105,-380) to node[bend right=0,auto,align=left] {$\mathit{}$} (wait);
\draw[arc,pos=0.8] (got_one) 
to node[bend right=0,auto,align=left] {\\$\mathit{1'(zero,}$ $\mathit{zero)}$ } (communication_channel);
\draw[arc,pos=0.2] (sender_buffer) to [bend right=0] (105,-12) to [bend right=0] (452,-12) to node[bend right=0,left,align=left,xshift=0mm] {$\mathit{}$ $\mathit{1'(x,}$ $\mathit{y,}$ $\mathit{z)}$ } (got_one);
\draw[arc,pos=-0.23] (got_one) to [bend right=0] (432,-42) to [bend right=0] (162,-42) to node[bend right=0,auto,align=left,xshift=23mm] {\\$\mathit{1'(y,}$ $\mathit{z,}$ $\mathit{one)}$ } (sender_buffer);
\draw[arc,pos=0.5] (got_one) to node[bend right=0,auto,align=left] {$\mathit{}$ } (first_qubit_sent);
\draw[inhibArc,pos=0.5, bend right=35] (first_qubit_sent)  to node[bend right=0,auto,align=left] {} (got_one);
\draw[arc,pos=0.5] (generated_bit) to node[bend right=0,auto,align=left] {$\mathit{}$ $\mathit{1'one}$ } (got_one);

\end{tikzpicture}
        }
        \caption{Send single bit}
        \label{fig:send_single_bit}
     \end{subfigure}
     \hfill
 \begin{subfigure}[b]{0.5\textwidth}
        \centering
        \scalebox{0.8}{
\begin{tikzpicture}[font=\scriptsize, xscale=0.45, yscale=0.45, x=1.33pt, y=1.33pt]
\tikzstyle{arc}=[->,>=stealth,thick]
\tikzstyle{transportArc}=[->,>=diamond,thick]
\tikzstyle{inhibArc}=[->,>=o,thick]
\tikzstyle{every place}=[minimum size=6mm,thick]
\tikzstyle{every transition} = [fill=black,minimum width=2mm,minimum height=5mm]
\tikzstyle{every token}=[fill=white,text=black]
\tikzstyle{sharedplace}=[place,minimum size=7.5mm,dashed,thin]
\tikzstyle{sharedtransition}=[transition, fill opacity=0, minimum width=3.5mm, minimum height=6.5mm,dashed]
\tikzstyle{urgenttransition}=[place,fill=white,minimum size=2.0mm,thin]
\tikzstyle{uncontrollabletransition}=[transition,fill=white,draw=black,very thick]
\tikzstyle{globalBox} = [draw,thick,align=left]
\node[place, label={[align=left,label distance=0cm]90:$\mathit{generated\_bit}$\\$\mathit{[bit]}$ }] at (330,-165) (generated_bit) {};
\node[sharedplace ] at (generated_bit.center) { };
\node[place, label={[align=left,label distance=0cm]90:$\mathit{start}$\\$\mathit{[dot]}$ \\$\mathit{1'dot}$ }] at (315,-90) (start) {};
\node at (315.0,-90.0){0.0};
\node[sharedplace ] at (start.center) { };
\node[place, label={[align=left,label distance=0cm,xshift=-4mm]90:$\mathit{sender\_buffer}$\\$\mathit{[buffer]}$ \\$\mathit{1'(zero,}$ $\mathit{zero,}$ $\mathit{zero)}$ }] at (630,-360) (sender_buffer) {};
\node at (630.0,-360.0){0.0};
\node[sharedplace ] at (sender_buffer.center) { };
\node[place, label={[align=left,label distance=0cm]90:$\mathit{wait}$\\$\mathit{[dot]}$ }] at (675,-90) (wait) {};
\node[sharedplace ] at (wait.center) { };
\node[place, label={[align=left,label distance=0cm,xshift=9mm]90:$\mathit{communication\_channel}$\\$\mathit{[bell\_pair]}$ }] at (345,-530) (communication_channel) {};
\node[sharedplace ] at (communication_channel.center) { };
\node[place, label={[align=left,label distance=0cm]90:$\mathit{first\_qubit\_sent}$\\$\mathit{[dot]}$ }] at (675,-165) (first_qubit_sent) {};
\node[sharedplace ] at (first_qubit_sent.center) { };
\node[place, label={[align=left,label distance=0cm]0:$\mathit{next\_bit}$\\$\mathit{[dot]}$ }] at (495,-255) (next_bit) {};
\node[place, label={[align=left,label distance=0cm]120:
$\mathit{[bit]}$ }] at (405,-360) (next_next_bit) {};
\node[place, label={[align=left,label distance=0cm]0:$\mathit{prepared\_bell\_pair}$\\$\mathit{[bell\_pair]}$ }] at (495,-450) (prepared_bell_pair) {};
\node[sharedplace ] at (prepared_bell_pair.center) { };
\node[transition, label={[align=left,label distance=0cm,yshift=0mm]90:$\mathit{ready}$\\$\mathit{constant(0)}$ }] at (495,-90) (ready) {};
\node[transition, label={[align=left,label distance=0cm,yshift=1mm]90:$\mathit{got\_one}$\\$\mathit{constant(1)}$ }] at (495,-165) (got_one) {};
\node[transition, label={[align=left,label distance=0cm,yshift=0mm,xshift=1mm]180:$\mathit{next}$\\$\mathit{constant(0)}$ }] at (405,-255) (next) {};
\node[transition, label={[align=left,label distance=0cm]90:$\mathit{next\_next}$\\$\mathit{constant(0)}$ }] at (495,-360) (next_next) {};
\node[transition, label={[align=left,label distance=0cm]20:$\mathit{send\_qubit}$\\$\mathit{constant(1)}$ }] at (495,-530) (send_qubit) {};
\draw[arc,pos=0.5] (wait) to node[bend right=0,auto,align=left] {$\mathit{}$ $\mathit{[f,f]}$ } (ready);
\draw[arc,pos=0.5] (ready) to node[bend right=0,auto,align=left] {} (start);
\draw[arc,pos=0.5] (generated_bit) to node[bend right=0,auto,align=left] {$\mathit{}$ $\mathit{1'one}$ } (got_one);
\draw[arc,pos=0.5] (generated_bit) to node[bend right=0,auto,align=left, xshift=-4mm, yshift=-2mm] {$\mathit{1'x}$ \\ \mbox{ } $\mathit{zero}$ $\mathit{\rightarrow}$ $\mathit{[1,1]}$ \\\mbox{  } \mbox{ } $\mathit{one}$ $\mathit{\rightarrow}$ $\mathit{[32,32]}$ } (next);
\draw[arc,pos=0.5, bend right=20] (generated_bit) .. controls (316.0,-434.0) .. (450,-404)  to node[bend right=10,auto,align=left,yshift=-4mm, xshift=-8mm] {$\mathit{}$ $\mathit{1'y}$ \\$\mathit{zero}$ $\mathit{\rightarrow}$ $\mathit{[1,1]}$ \\$\mathit{one}$ $\mathit{\rightarrow}$ $\mathit{[32,32]}$ } (next_next);
\draw[arc,pos=0.5,yshift=0mm] (send_qubit) to node[bend right=0,auto,above,align=left,yshift=0mm] {\\$\mathit{1'(x,}$ $\mathit{y)}$ } (communication_channel);
\draw[arc,pos=0.5] (first_qubit_sent) to [bend right=0] (675,-530) to node[bend right=0,auto,align=left] {$\mathit{}$  } (send_qubit);
\draw[transportArc,pos=0.5] (start) to node[bend right=0,auto,align=left] {$\mathit{}$ $\mathit{}$   } (got_one);
\draw[transportArc,pos=0.5] (got_one) to node[bend right=0,auto,align=left] {$\mathit{}$ } (wait);
\draw[arc,pos=0.5] (got_one) to node[bend right=0,auto,align=left] { } (next_bit);
\draw[arc,pos=0.5] (next_bit) to node[bend right=0,auto,align=left] { } (next);
\draw[arc,pos=0.5] (next) to node[bend right=0,auto,align=left] {\\$\mathit{1'x}$ } (next_next_bit);
\draw[arc,pos=0.5] (next_next_bit) to node[bend right=0,auto,align=left] {$\mathit{}$ $\mathit{1'x}$  } (next_next);
\draw[arc,pos=0.7] (next_next) to node[bend right=0,auto,align=left] {\\$\mathit{1'(x,}$ $\mathit{y)}$ } (prepared_bell_pair);
\draw[arc,pos=0.4,xshift=-2mm] (prepared_bell_pair) to node[bend right=0,auto,align=left] {$\mathit{}$ $\mathit{1'(x,}$ $\mathit{y)}$ } (send_qubit);
\draw[arc,pos=0.5,bend left=40] (sender_buffer)  to node[bend right=0,auto,align=left,yshift=1.5mm] {$\mathit{}$ $\mathit{1'(u,}$ $\mathit{v,}$ $\mathit{w)}$ } (next_next);
\draw[arc,pos=0.5] (next_next) to node[bend right=0,auto,below,yshift=1.5mm,align=left] {\\$\mathit{1'(one,}$ $\mathit{x,}$ $\mathit{y)}$ } (sender_buffer);
\draw[arc,pos=0.5, bend right=20] (first_qubit_sent) to node[bend right=0,auto,align=left] {$\mathit{}$  } (got_one);
\draw[arc,pos=0.5,bend right=20] (got_one)  to node[bend right=0,auto,align=left] {$\mathit{}$  } (first_qubit_sent);

\end{tikzpicture}
        }
        \caption{Send three bits}
        \label{fig:send_three_bits}
     \end{subfigure}   
    \caption{BSDC bit generation and sender behaviour}
    \label{fig:sender}
\end{figure}

The protocol  relies both on the strict timing constraints
(length of the frequency slot vs. the duration of the application of quantum gates,
creation of the EPR pair, and the duration of measurements)
as well as on probabilistic aspects due to
decoherence (loss of entanglement between the two qubits). This 
means that with the increasing time delay between creating of the EPR pair
until the measurement of the qubits is performed, the receiver has 
a decreasing probability that it will decode correctly the two encoded classical bits.

In order to use realistic timing and decoherence information, we
overtook this information from an existing IBM quantum computer~\cite{IBMData},
where the application of 1-qubit gates takes $32$ ns (nano seconds),
2-qubit gates require 138 ns, a qubit measurement takes $2400$ ns and the
reset of the qubits before every round has estimated duration of $480$ ns.

Similarly, to obtain hardware-realistic decoherence times for the generated
EPR pair which can be in one of the possible four Bell states~\cite{NielsenChuang2010},
we model the mixed states of the EPR pair as a density matrix~\cite{Watrous2018} that represents
the probabilistic distribution of the four Bell states where the 
likehood of obtaining a wrong measurement increases over time
with the coefficients T1=$343.25$ $\mu$s for relaxation time and
T2=$263.65$ $\mu$s for dephasing, also taken from~\cite{IBMData}. This is also
known as Pauli-twirl~\cite{Geller2013}.
We sampled a series of 1000 decoherence times from the obtained density matrix (computed
separately in Python) and created a corresponding custom distribution that is imported
to TAPAAL. Should the live-time
of the created EPR pair exceed the sampled value from this custom distribution,
we assume that instead of reading the correct Bell pair, during the measurement we
obtain any of the four possible Bell states with uniform probability.

The sender part of our TAPAAL model of the BSDC protocol is depicted
in Figure~\ref{fig:sender}. The model uses the colored extension of the model
and relies on the notion of shared places, already introduced in the previous case
studies. The color type \emph{bit} represents two possible values of a classical
bit information and Bell pair is represented as a product of two classical bits, 
encoding the four possible results of the measurement.
We also use a buffer color type of size three 
where we remember the last three sent/received classical bits and use
these buffers to check the error rates in case that the~sent/received
bits disagree due to timing conflicts or decoherence. The model
also uses the constant $f$ that is the duration of the time slot
of the protocol (in nano seconds) and we vary this constant in order to
explore the behaviour of the protocol under different sending frequencies.
Figure~\ref{fig:bit_generator} shows a simple way how we model
the incoming bits that should be transmitted. Whenever the shared place
\emph{generated\_bit} is empty, a bit \emph{zero} or \emph{one}
is immediately inserted to the place, assuming that the bits arrive
with uniform probability (both transitions in the component have the default
weight $1$). The two transitions are blocked until the colored token
from the place \emph{generated\_bit} is consumed by the sender.

This can happen in the component in Figure~\ref{fig:send_single_bit}
where we model the sending of the bits $0$ and of the first bit $1$.
Initially, the component has a token  of age $0.0$ in the place \emph{start}. In case
that the arriving bit in the place \emph{generated\_bit} has
the value \emph{zero}, we shall fire the transition \emph{got\_zero}
that has the duration of $1$ ns, consumes the token from \emph{generated\_bit}
and places the control token to the place \emph{wait}. As we use the transport
arcs here, the age of the control token is preserved and we wait until
the end of the frequency slot, after which the token is placed
again into \emph{start} by firing the transition \emph{ready}. At the same
time, we read the three bits from the place \emph{sender\_buffer}, shift
the bits by one position and remember on the last position in the buffer that
we just dispatched the bit $0$. In this case, no qubit is passed
to the place \emph{communication\_channel}. 
If on the other hand \emph{generated\_bit} contains a token with 
color \emph{one}, we fire the transition
\emph{got\_one} which remembers in \emph{sender\_buffer} 
that we are dispatching the bit $1$, it prepares the EPR pair
and sends the first qubit to the quantum communication channel (the values
of the qubit are irrelevant in this phase). The transition also
places a token to the place
\emph{first\_qubit\_sent} which makes sure that once the second bit $1$
arrives, we have to perfrom the encoding process described in Figure~\ref{fig:send_three_bits} instead. The duration of the transition
\emph{got\_one} is $650$ ns, drawn from normal distribution with variance of $50$.
Here we have to perfrom the reset of two qubits (480 ns), followed by 
the application of the Hadamard gate H (32 ns) on the first qubit followed
by the application of CNOT on both qubits (138 ns) in order to prepare the EPR pair.

Upon the arrival of the second bit $1$, we instead perform the encoding
procedure described in Figure~\ref{fig:send_three_bits}. After the second generated
bit $1$ is consumed by the transition \emph{got\_one} (which is possible
only if the place \emph{first\_quibit\_sent} has a token), we place a token
to the place \emph{wait} that as before serves as a timer in order to keep
the protocol in sync at frequency $f$. At the same time, we  fire a sequence
of two transitions \emph{next} and \emph{next\_next} that consume
the next two input bits from the place \emph{generated\_bit} and encode
them into a qubit that is moved to the place \emph{prepared\_bell\_pair}.
The application of the X and Z gates in case that the input bits are $1$
causes a delay of $32$ ns whereas the reading of the bit $0$ has delay of $1$ ns
(application of the identity gate). This is depicted by the color specific
intervals on the input arcs. Upon consuming the three bits from 
\emph{generated\_bit}, the values are remembered in \emph{sender\_buffer}
and the generated qubit is forward to \emph{communication\_channel} while
as the same time consuming the token from \emph{first\_qubit\_sent}.
We shall discuss how we model decoherence of the EPR pair later on.

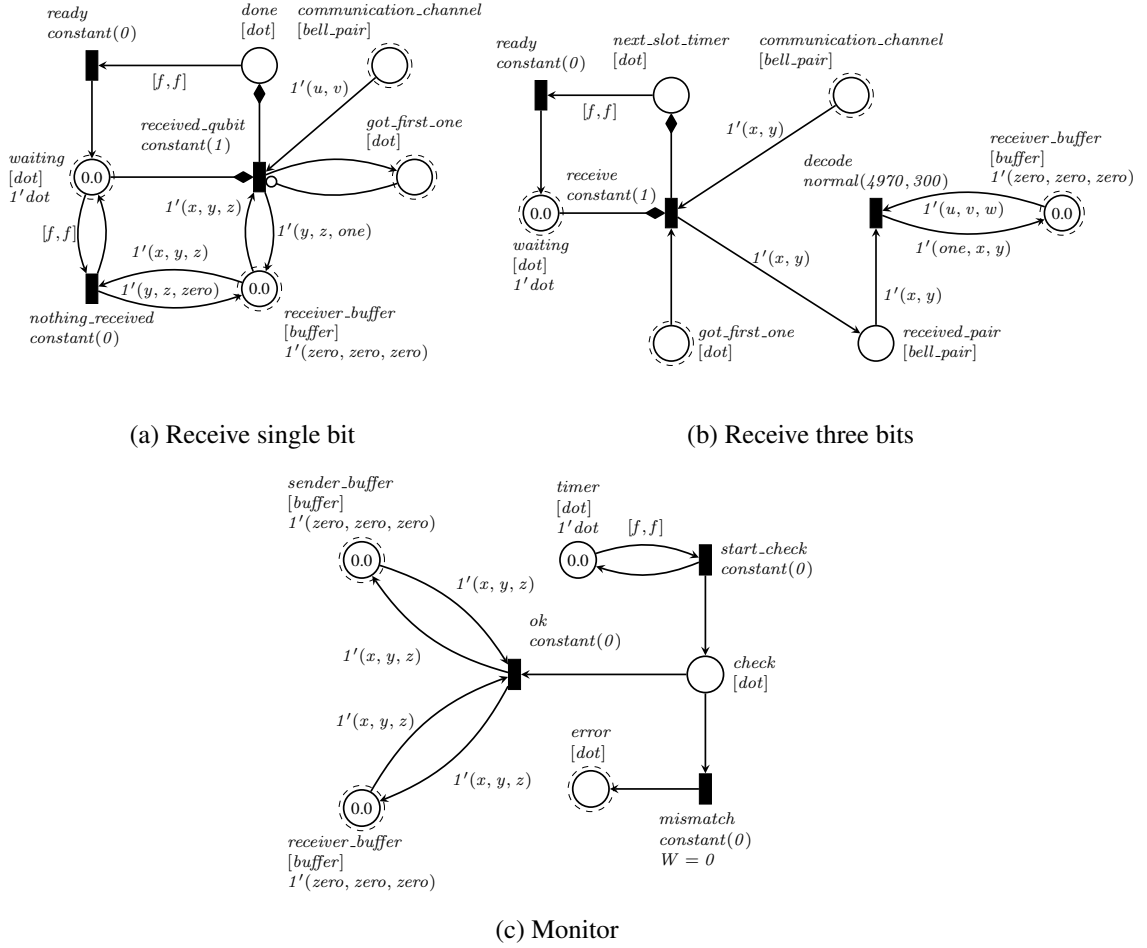
\begin{figure}[t]
    \centering
    \begin{subfigure}[b]{0.42\textwidth}
        \centering
        \scalebox{0.78}{
\begin{tikzpicture}[font=\scriptsize, xscale=0.45, yscale=0.45, x=1.33pt, y=1.33pt]
\tikzstyle{arc}=[->,>=stealth,thick]
\tikzstyle{transportArc}=[->,>=diamond,thick]
\tikzstyle{inhibArc}=[->,>=o,thick]
\tikzstyle{every place}=[minimum size=6mm,thick]
\tikzstyle{every transition} = [fill=black,minimum width=2mm,minimum height=5mm]
\tikzstyle{every token}=[fill=white,text=black]
\tikzstyle{sharedplace}=[place,minimum size=7.5mm,dashed,thin]
\tikzstyle{sharedtransition}=[transition, fill opacity=0, minimum width=3.5mm, minimum height=6.5mm,dashed]
\tikzstyle{urgenttransition}=[place,fill=white,minimum size=2.0mm,thin]
\tikzstyle{uncontrollabletransition}=[transition,fill=white,draw=black,very thick]
\tikzstyle{globalBox} = [draw,thick,align=left]
\node[place, label={[align=left,label distance=0cm]90:$\mathit{communication\_channel}$\\$\mathit{[bell\_pair]}$ }] at (480,-120) (communication_channel) {};
\node[sharedplace ] at (communication_channel.center) { };
\node[place, label={[align=left,label distance=0cm]180:$\mathit{waiting}$\\$\mathit{[dot]}$ \\$\mathit{1'dot}$ }] at (240,-210) (waiting) {};
\node at (240.0,-210.0){0.0};
\node[sharedplace ] at (waiting.center) { };
\node[place, label={[align=left,label distance=0cm]330:$\mathit{receiver\_buffer}$\\$\mathit{[buffer]}$ \\$\mathit{1'(zero,}$ $\mathit{zero,}$ $\mathit{zero)}$ }] at (375,-300) (receiver_buffer) {};
\node at (375.0,-300.0){0.0};
\node[sharedplace ] at (receiver_buffer.center) { };
\node[place, label={[align=left,label distance=0cm]90:$\mathit{done}$\\$\mathit{[dot]}$}] at (375,-120) (done) {};
\node[place, label={[align=left,label distance=0cm]90:$\mathit{got\_first\_one}$\\$\mathit{[dot]}$ }] at (500,-210) (got_first_one) {};
\node[sharedplace ] at (got_first_one.center) { };
\node[transition, label={[align=left,label distance=0cm]270:$\mathit{nothing\_received}$\\$\mathit{constant(0)}$ }] at (240,-300) (nothing_received) {};
\node[transition, label={[align=left,label distance=0cm]110:$\mathit{received\_qubit}$\\$\mathit{constant(1)}$ }] at (375,-210) (received_qubit) {};
\node[transition, label={[align=left,label distance=0cm]90:$\mathit{ready}$\\$\mathit{constant(0)}$ }] at (240,-120) (ready) {};
\draw[arc,pos=0.5, bend right=20] (waiting) to node[bend right=0,auto,align=right,left,xshift=2mm] {$\mathit{[f,f]}$ } (nothing_received);
\draw[arc,pos=0.5, bend right=20] (nothing_received)   to node[bend right=0,auto,align=left] {\\$\mathit{1'(y,}$ $\mathit{z,}$ $\mathit{zero)}$ } (receiver_buffer);
\draw[arc,pos=0.5, bend right=20] (nothing_received)  to node[bend right=0,auto,align=left] { } (waiting);
\draw[arc,pos=0.5, bend right=20] (receiver_buffer)  to node[bend right=0,auto,above,align=left] {$\mathit{}$ $\mathit{1'(x,}$ $\mathit{y,}$ $\mathit{z)}$ } (nothing_received);
\draw[arc,pos=0.3] (communication_channel) to node[bend right=0,auto,align=left,above,xshift=-4mm] {$\mathit{}$ $\mathit{1'(u,}$ $\mathit{v)}$ } (received_qubit);
\draw[transportArc,pos=0.5] (waiting) to node[bend right=0,auto,align=left] {$\mathit{}$  } (received_qubit);
\draw[transportArc,pos=0.5] (received_qubit) to node[bend right=0,auto,align=left] {$\mathit{}$ $\mathit{}$  } (done);
\draw[arc,pos=0.5, bend left=0] (done) to node[bend right=0,auto,align=right] {$\mathit{}$ $\mathit{[f,f]}$ } (ready);
\draw[arc,pos=0.5] (ready) to  node[bend right=0,auto,align=left] { } (waiting);
\draw[arc,pos=0.6, bend left=20] (receiver_buffer)   to node[bend right=0,auto,align=left,xshift=1.5mm] {$\mathit{}$ $\mathit{1'(x,}$ $\mathit{y,}$ $\mathit{z)}$ } (received_qubit);
\draw[arc,pos=0.5, bend left=20] (received_qubit)   to node[bend right=0,auto,align=left, xshift=-1mm] {$\mathit{1'(y,}$ $\mathit{z,}$ $\mathit{one)}$ } (receiver_buffer);
\draw[arc,pos=0.5,bend left=20] (received_qubit)  to node[bend right=0,auto,align=left] { } (got_first_one);
\draw[inhibArc,pos=0.5, bend left=20] (got_first_one)   to node[bend right=0,auto,align=left] {} (received_qubit);

\end{tikzpicture}
        }
        \vspace{0mm}
        \caption{Receive single bit}
        \label{fig:receive_single_bit}
     \end{subfigure} \hspace{-3mm}
     \begin{subfigure}[b]{0.54\textwidth}
        \centering
        \scalebox{0.78}{
\begin{tikzpicture}[font=\scriptsize, xscale=0.45, yscale=0.45, x=1.33pt, y=1.33pt]
\tikzstyle{arc}=[->,>=stealth,thick]
\tikzstyle{transportArc}=[->,>=diamond,thick]
\tikzstyle{inhibArc}=[->,>=o,thick]
\tikzstyle{every place}=[minimum size=6mm,thick]
\tikzstyle{every transition} = [fill=black,minimum width=2mm,minimum height=5mm]
\tikzstyle{every token}=[fill=white,text=black]
\tikzstyle{sharedplace}=[place,minimum size=7.5mm,dashed,thin]
\tikzstyle{sharedtransition}=[transition, fill opacity=0, minimum width=3.5mm, minimum height=6.5mm,dashed]
\tikzstyle{urgenttransition}=[place,fill=white,minimum size=2.0mm,thin]
\tikzstyle{uncontrollabletransition}=[transition,fill=white,draw=black,very thick]
\tikzstyle{globalBox} = [draw,thick,align=left]
\node[place, label={[align=left,label distance=0cm]90:$\mathit{communication\_channel}$\\$\mathit{[bell\_pair]}$ }] at (370,-85) (communication_channel) {};
\node[sharedplace ] at (communication_channel.center) { };
\node[place, label={[align=left,label distance=0cm]270:$\mathit{waiting}$\\$\mathit{[dot]}$ \\$\mathit{1'dot}$ }] at (120,-180) (waiting) {};
\node at (120.0,-180.0){0.0};
\node[sharedplace ] at (waiting.center) { };
\node[place, label={[align=left,label distance=0cm]90:$\mathit{receiver\_buffer}$\\$\mathit{[buffer]}$ \\$\mathit{1'(zero,}$ $\mathit{zero,}$ $\mathit{zero)}$ }] at (540,-180) (receiver_buffer) {};
\node at (540.0,-180.0){0.0};
\node[sharedplace ] at (receiver_buffer.center) { };
\node[place, label={[align=left,label distance=0cm]0:$\mathit{received\_pair}$\\$\mathit{[bell\_pair]}$ }] at (390,-285) (received_pair) {};
\node[place, label={[align=left,label distance=0cm]0:$\mathit{got\_first\_one}$\\$\mathit{[dot]}$ }] at (225,-285) (got_first_one) {};
\node[sharedplace ] at (got_first_one.center) { };
\node[place, label={[align=left,label distance=0cm]90:$\mathit{next\_slot\_timer}$\\$\mathit{[dot]}$ }] at (225,-85) (next_slot_timer) {};
\node[transition, label={[align=left,label distance=0cm]170:$\mathit{receive}$\\$\mathit{constant(1)}$ }] at (225,-180) (received_qubit) {};
\node[transition, label={[align=left,label distance=0cm]90:$\mathit{decode}$\\$\mathit{normal(4970,300)}$ }] at (390,-180) (decode) {};
\node[transition, label={[align=left,label distance=0cm]90:$\mathit{ready}$\\$\mathit{constant(0)}$ }] at (120,-85) (ready) {};
\draw[arc,pos=0.4] (communication_channel) to node[bend right=0,auto,align=left,above,xshift=-3mm] {$\mathit{}$ $\mathit{1'(x,}$ $\mathit{y)}$ } (received_qubit);
\draw[arc,pos=0.5] (got_first_one) to node[bend right=0,auto,align=left] { } (received_qubit);
\draw[arc,pos=0.5, bend right=20] (receiver_buffer) to node[bend right=0,auto,align=left] {$\mathit{}$ $\mathit{1'(u,}$ $\mathit{v,}$ $\mathit{w)}$ } (decode);
\draw[arc,pos=0.5, bend right=20] (decode)  to node[bend right=0,auto,below,align=left,yshift=1.5mm] {\\$\mathit{1'(one,}$ $\mathit{x,}$ $\mathit{y)}$ } (receiver_buffer);
\draw[arc,pos=0.5] (next_slot_timer) to node[bend right=0,auto,align=left] {$\mathit{}$ $\mathit{[f,f]}$ } (ready);
\draw[arc,pos=0.5] (ready) to node[bend right=0,auto,align=left] { } (waiting);
\draw[transportArc,pos=0.5] (waiting) to node[bend right=0,auto,align=left] {  } (received_qubit);
\draw[transportArc,pos=0.5] (received_qubit) to node[bend right=0,auto,align=left] {$\mathit{}$ } (next_slot_timer);
\draw[arc,pos=0.5] (received_qubit) to node[bend right=0,auto,above,align=left,xshift=2mm] {\\$\mathit{1'(x,}$ $\mathit{y)}$ } (received_pair);
\draw[arc,pos=0.3] (received_pair) to node[bend right=0,auto,align=left,right,xshift=-1mm] {$\mathit{}$ $\mathit{1'(x,}$ $\mathit{y)}$ } (decode);
\end{tikzpicture}
        }
        \vspace{0cm}
        \caption{Receive three bits}
        \label{fig:receive_three_bits}
     \end{subfigure}
     
 \begin{subfigure}[b]{0.6\textwidth}
        \centering
        \scalebox{0.8}{
\begin{tikzpicture}[font=\scriptsize, xscale=0.45, yscale=0.45, x=1.33pt, y=1.33pt]
\tikzstyle{arc}=[->,>=stealth,thick]
\tikzstyle{transportArc}=[->,>=diamond,thick]
\tikzstyle{inhibArc}=[->,>=o,thick]
\tikzstyle{every place}=[minimum size=6mm,thick]
\tikzstyle{every transition} = [fill=black,minimum width=2mm,minimum height=5mm]
\tikzstyle{every token}=[fill=white,text=black]
\tikzstyle{sharedplace}=[place,minimum size=7.5mm,dashed,thin]
\tikzstyle{sharedtransition}=[transition, fill opacity=0, minimum width=3.5mm, minimum height=6.5mm,dashed]
\tikzstyle{urgenttransition}=[place,fill=white,minimum size=2.0mm,thin]
\tikzstyle{uncontrollabletransition}=[transition,fill=white,draw=black,very thick]
\tikzstyle{globalBox} = [draw,thick,align=left]
\node[place, label={[align=left,label distance=0cm]90:$\mathit{sender\_buffer}$\\$\mathit{[buffer]}$ \\$\mathit{1'(zero,}$ $\mathit{zero,}$ $\mathit{zero)}$ }] at (150,-105) (sender_buffer) {};
\node at (150.0,-105.0){0.0};
\node[sharedplace ] at (sender_buffer.center) { };
\node[place, label={[align=left,label distance=0cm]270:$\mathit{receiver\_buffer}$\\$\mathit{[buffer]}$ \\$\mathit{1'(zero,}$ $\mathit{zero,}$ $\mathit{zero)}$ }] at (150,-300) (receiver_buffer) {};
\node at (150.0,-300.0){0.0};
\node[sharedplace ] at (receiver_buffer.center) { };
\node[place, label={[align=left,label distance=0cm]0:$\mathit{check}$\\$\mathit{[dot]}$ }] at (420,-195) (check) {};
\node[place, label={[align=left,label distance=0cm]90:$\mathit{error}$\\$\mathit{[dot]}$ }] at (330,-285) (error) {};
\node[sharedplace ] at (error.center) { };
\node[place, label={[align=left,label distance=0cm]90:$\mathit{timer}$\\$\mathit{[dot]}$ \\$\mathit{1'dot}$ }] at (320,-105) (timer) {0.0};
\node[transition, label={[align=left,label distance=0cm]70:$\mathit{ok}$\\$\mathit{constant(0)}$ }] at (270,-195) (ok) {};
\node[transition, label={[align=left,label distance=0cm]270:$\mathit{mismatch}$\\$\mathit{constant(0)}$ \\$\mathit{W=0}$ }] at (420,-285) (mismatch) {};
\node[transition, label={[align=left,label distance=0cm]0:$\mathit{start\_check}$\\$\mathit{constant(0)}$ }] at (420,-105) (start_check) {};
\draw[arc,pos=0.5] (check) to node[bend right=0,auto,align=left] {$\mathit{}$  } (ok);
\draw[arc,pos=0.5,bend left=20] (sender_buffer) to node[bend right=0,auto,align=left,yshift=0mm,xshift=-2mm] {$\mathit{}$ $\mathit{1'(x,}$ $\mathit{y,}$ $\mathit{z)}$ } (ok);
\draw[arc,pos=0.5,bend left=20] (receiver_buffer) to node[bend right=0,auto,align=left,yshift=-3mm] {$\mathit{}$ $\mathit{1'(x,}$ $\mathit{y,}$ $\mathit{z)}$ } (ok);
\draw[arc,pos=0.5] (check) to node[bend right=0,auto,align=left] {} (mismatch);
\draw[arc,pos=0.5] (mismatch) to node[bend right=0,auto,align=left] { } (error);
\draw[arc,pos=0.5,bend left=20] (ok) to node[bend right=0,auto,align=left,xshift=-2mm] {\\$\mathit{1'(x,}$ $\mathit{y,}$ $\mathit{z)}$ } (receiver_buffer);
\draw[arc,pos=0.5,bend left=20] (ok)  to node[bend right=0,auto,align=left,yshift=1mm] {\\$\mathit{1'(x,}$ $\mathit{y,}$ $\mathit{z)}$ } (sender_buffer);
\draw[arc,pos=0.5] (start_check) to node[bend right=0,auto,align=left] { } (check);
\draw[arc,pos=0.5,bend left=20] (timer)  to node[bend right=0,auto,align=left] {$\mathit{[f,f]}$ } (start_check);
\draw[arc,pos=0.5,bend left=20] (start_check)   to node[bend right=0,auto,align=left] {} (timer);

\end{tikzpicture}
        }
        \vspace{0cm}
        \caption{Monitor}
        \label{fig:monitor}
     \end{subfigure}
     
    \caption{BSDC receiver behaviour and monitor}
    \label{fig:receiver}
\end{figure}

Receiver's behaviour as well as monitoring of the sent/received information bits
is shown in Figure~\ref{fig:receiver}. Initially, the receiver in Figure~\ref{fig:receive_single_bit}
has a token of age $0.0$ in place \emph{waiting}. Should the length
of the frequency slot be reached without any qubit arriving
on \emph{communication\_channel}, it fires the transition
\emph{nothing\_received}, resets the age of the token to $0$ and
updates the content of \emph{receiver\_buffer} by remembering that the last
received bit is $0$. Upon the arrival of the first qubit, the receiver
fires the transition \emph{received\_qubit}, interprets it as receiving
the bit $1$, and places a token to \emph{done} where it waits to the next 
time slot. It also places a token into the place \emph{got\_first\_one}
where it remembers that on the arrival of the next qubit, it should decode
in which Bell state the two qubits were prepared. This is done 
in Figure~\ref{fig:receive_three_bits} where the second received qubit is moved
to the place \emph{received\_pair} by firing the transition \emph{receive}
and at the same time removing a token from \emph{got\_first\_one} and
keeping the frequency timer by moving a token to \emph{done} while preserving
its age. Finally, receiver applies the CNOT gate on its two qubits ($138$ ns),
performs the H gate ($32$ ns) and then two measurements ($2400$ ns each),
resulting in the total processing time of $4970$ ns, drawn from the
normal distribution with variance of $300$ by firing the transition
\emph{decode} and remembering the measured values in the receiver's buffer.

Next, in Figure~\ref{fig:monitor} we check at the end of each time slot
whether the buffered bits in the sender's and receiver's buffer match.
Should this be the case then the token in the place \emph{check} is consumed
by firing the transition \emph{ok}, otherwise we fire the transition
\emph{mismatch} and add a token to the place \emph{error} where we count
the number of encountered errors. Note that the transition \emph{mismatch}
has weight $0$ and will only fire if \emph{ok} is not enabled. Also,
it is enough to consider buffers of size three only as this is the maximum number of classical bits that can be sent in each time slot. In the actual
model of the protocol, the monitor is slightly more complicated as it counts
the exact number of mismatched bits in the two buffers, by checking a mismatch
at each position of the buffer.

\begin{figure}[t]
    \centering
    \scalebox{0.8}{
\begin{tikzpicture}[font=\scriptsize, xscale=0.45, yscale=0.45, x=1.33pt, y=1.33pt]
\tikzstyle{arc}=[->,>=stealth,thick]
\tikzstyle{transportArc}=[->,>=diamond,thick]
\tikzstyle{inhibArc}=[->,>=o,thick]
\tikzstyle{every place}=[minimum size=6mm,thick]
\tikzstyle{every transition} = [fill=black,minimum width=2mm,minimum height=5mm]
\tikzstyle{every token}=[fill=white,text=black]
\tikzstyle{sharedplace}=[place,minimum size=7.5mm,dashed,thin]
\tikzstyle{sharedtransition}=[transition, fill opacity=0, minimum width=3.5mm, minimum height=6.5mm,dashed]
\tikzstyle{urgenttransition}=[place,fill=white,minimum size=2.0mm,thin]
\tikzstyle{uncontrollabletransition}=[transition,fill=white,draw=black,very thick]
\tikzstyle{globalBox} = [draw,thick,align=left]
\node[place, label={[align=left,label distance=0cm]90:$\mathit{first\_qubit\_sent}$\\$\mathit{[dot]}$ }] at (90,-75) (first_qubit_sent) {};
\node[sharedplace ] at (first_qubit_sent.center) { };
\node[place, label={[align=left,label distance=0cm,xshift=3mm]90:$\mathit{prepared\_bell\_pair}$\\$\mathit{[bell\_pair]}$ }] at (660,-75) (prepared_bell_pair) {};
\node[sharedplace ] at (prepared_bell_pair.center) { };
\node[place, label={[align=left,label distance=0cm]90:$\mathit{decohered}$\\$\mathit{[dot]}$ }] at (405,-75) (decohered) {};
\node[transition, label={[align=left,label distance=0cm]90:$\mathit{decohere}$\\$\mathit{custom(Pauli)}$ }] at (255,-75) (decohere) {};
\node[transition, label={[align=left,label distance=0cm]90:$\mathit{introduce\_error}$\\$\mathit{constant(0)}$ }] at (525,-75) (introduce_error) {};
\draw[arc,pos=0.5, bend right=20] (first_qubit_sent)  to node[bend right=0,auto,align=left] {$\mathit{}$ } (decohere);
\draw[arc,pos=0.5,bend right=20] (decohere) to node[bend right=0,auto,align=left] {} (decohered);
\draw[arc,pos=0.5] (decohered) to node[bend right=0,auto,align=left] {$\mathit{}$ } (introduce_error);
\draw[arc,pos=0.5, bend right=20,yshift=0mm] (prepared_bell_pair)   to node[bend right=0,auto,align=left,yshift=0.5mm] {$\mathit{}$ $\mathit{1'(u,}$ $\mathit{v)}$ } (introduce_error);
\draw[arc,pos=0.5,bend right=20] (introduce_error)   to node[bend right=0,auto,below,align=left,yshift=1.5mm] {\\$\mathit{1'(x,}$ $\mathit{y)}$ } (prepared_bell_pair);
\draw[inhibArc,pos=0.5,bend right=20] (decohered)   to node[bend right=0,auto,align=left] {} (decohere);
\draw[arc,pos=0.5, bend right=20] (decohere)   to node[bend right=0,auto,align=left] { } (first_qubit_sent);

\end{tikzpicture}
    }
    \caption{Decoherence}
    \label{fig:decoherence}
\end{figure}
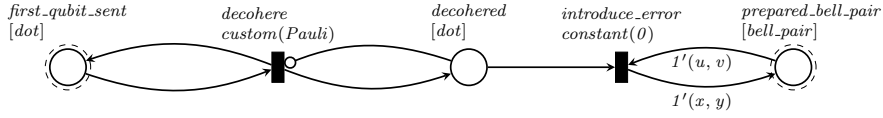

Finally, Figure~\ref{fig:decoherence} shows how we model the decoherence
of the EPR pair over time. Once the first quibit is sent and a token
is added to the place \emph{first\_quibit\_sent}, the transition
\emph{decohere} will sample its firing time from the custom distribution called \emph{Pauli} that was mentioned before. If the second qubit is dispatched before
the sampled time delay, the transition \emph{decohere} becomes disabled
and the receiver will recover the correct bits encoded to second qubit by
the sender. Should the decoherence time be reached, the transition \emph{decohere}
fires and enables the transition \emph{introduce\_error} that will
consume the prepared Bell pair $(u,v)$ and replace it with another
of the possible Bell pairs (each of them with $0.25$ probability). This can
introduce up to 2 bit-flips once the receiver decodes the Bell pair.


\begin{figure}[t]
\centering
\scalebox{0.9}{
\begin{tabular}{|r||r | r | r | r | r || r | r |}
\hline
Frequency & 1-Error &  2-Error &  3-Error &  4-Error & 5-Error & Throughput & Corrected \\ \hline\hline 
5 $\mu$s &  0.994 &  0.986 &  0.970 & 0.941 & 0.890 & 200 kbit/s &  64 kbit/s \\ \hline
5.5 $\mu$s &  0.492 &  0.356 &  0.236 & 0.146 & 0.082 & 182 kbit/s & 83 kbit/s \\ \hline
6 $\mu$s &  0.290 &  0.130 &  0.031 & 0.010 & 0.002 & 167 kbit/s & 101 kbit/s \\ \hline
6.5 $\mu$s &  0.311 &  0.138 &  0.031 & 0.009 & 0.002 & 154 kbit/s & 93 kbit/s \\ \hline
7 $\mu$s &  0.350 &  0.156 &  0.036 & 0.010 & 0.002 & 143 kbit/s &  86 kbit/s \\ \hline
7.5 $\mu$s &  0.382 &  0.176 &  0.045 & 0.013 & 0.003 & 133 kbit/s & 80 kbit/s \\ \hline
8 $\mu$s &  0.401 &  0.189 &  0.051 & 0.015 & 0.003 & 125 kbit/s & 75 kbit/s \\ \hline
8.5 $\mu$s &  0.426 &  0.205 &  0.057 & 0.017 & 0.003 & 118 kbit/s & 71 kbit/s \\ \hline
9 $\mu$s &  0.451 &  0.220 &  0.063 & 0.019 & 0.004 & 111 kbit/s & 67 kbit/s \\ \hline
9.5 $\mu$s &  0.477 &  0.234 &  0.069 & 0.021 & 0.004 & 105 kbit/s & 63 kbit/s \\ \hline

\end{tabular}
}
\caption{Error probabilities for sending 50 bits of information for a given number of bit-flips and throughput before and after applying error correction}
\label{fig:results}
\end{figure}

Based on this TAPAAL SMC model, we now run a number of SMC queries that
estimate the probability that while sending 50 classical bits,
there are at least $X$ bit flips on the receiver side (query called
$X$-Error that estimates $\prob({\tt F \ error } \geq X)$). These queries for up to $5$ bit-flips are evaluated for different frequencies 
(durations of the time slot), with precision $0.001$ and confidence 
of $95$\%. With this high precision and timing in nano seconds, 
evaluation of such a query takes between 10 to 50 minutes on a Mac Studio with
Apple M2 Max 12‑core CPU. The memory requirements for SMC are negligible
(less than 50 MB).  The results are displayed in Figure~\ref{fig:results}
and the lowest (best) values of error probabilities are observed
for the frequency slot duration of $6$ $\mu$s. If the frequencies are shorter,
there is an increasing risk of not finishing all the quantum operations during
the given time slot, causing a significant increase of the error rates. On the
other hand, slower frequencies suffer from increased number of errors due
to the decoherence of the Bell pair.

The next column in Figure~\ref{fig:results} shows the computed
idealized throughput (number of delivered classical bits per second)
of the protocol, disregarding the possible bit-flips. As bit-flips
are not avoidable, we can apply classical error correction.
To do so, we formulate for each of the analyzed frequencies
a~qualitative hypothesis for sending 50 classical bits
\begin{center}
${\prob({\tt F \ error } \geq X) \geq 0.001}$
\end{center} 
in order to determine the smallest number of bit-flips $X$ that 
can happen with probability strictly lower than $0.001$.
To identify such $X$, we run the query for different values
of $X$ (using e.g. the bisection method) until we find the smallest number of bit-flips for
which the query evaluates to false. These tests are run with
the false positive/negative parameters set to $0.001$ and indifference region width of $0.0002$, and testing the qualitative
hypothesis takes typically less than a minute.

\begin{figure}[t]
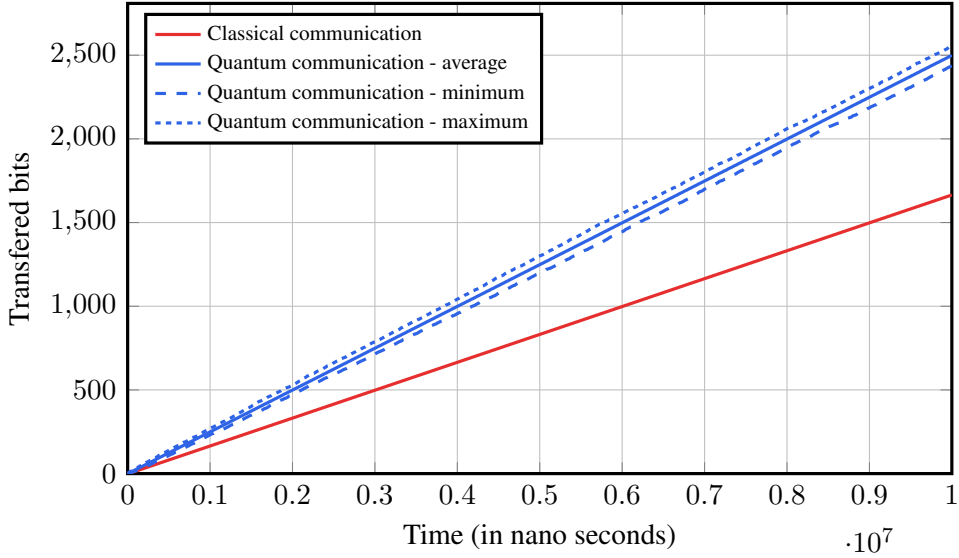

\begin{center}
\include{plots/BSDC-6000-bit-throughput}
\end{center}
\vspace{-10mm}
\caption{Number of transferred bits at 6 $\mu$s frequency}
\label{fig:transferred}
\end{figure}

In order to guarantee that a message with $50$ classical bits is received incorrectly with a probability strictly less than $0.001$,
we need to be resilient for $25$ bit-flips for $5$ $\mu$s
frequency, $12$ bit-flips for $5.5$ $\mu$s frequency,
$6$ bit-flips for the remaining frequencies.
In order to achieve this, we have to use classical error correction codes
which guarantee twice as high Hamming distance as the number
of possible bit-flips~\cite{MacWilliamsSloane1977}; this translates to adding additional
$107$ error correcting bits for $5$ $\mu$s frequency,
$59$ bits for $5.5$ $\mu$s and $33$ extra bits for the remaining frequencies.
The last column in the table from Figure~\ref{fig:results} displays
the corrected throughput after applying the error correction codes. As expected, the highest throughput of $101$ kbit/s after applying the error correction is achieved for the frequency of $6$ $\mu$s.

We conclude this case study by adding observations (similarly as in the manufacturing process case study) to our TAPAAL SMC query, which monitor
how many
bits can be communicated using classical communication
(where we send one classical bit per time slot) compared to
the number of transferred classical bits using the BSDC quantum
protocol (where we send at most one qubit per
time slot). A TAPAAL generated plot comparing these two
methods is shown in Figure~\ref{fig:transferred}. As
the number of transferred bits in the quantum protocol
depends on the actual distribution of $0$ and $1$ in the input
message, we depict the average as well as minimum/maximum number of transferred bits for a~given time duration. The plot shows
on average 48\% increase in the number of transferred bits in the quantum protocol compared to the classical communication using the same frequency time slot of $6$ $\mu$s when sending a single bit per time slot.

\section{Conclusion}

We proposed, to the best of our knowledge, the first stochastic semantics to the popular
timed-arc Petri net model. The semantics is weak and race-based, implying that an enabled
transition may not necessarily fire within its firing interval but can instead delay out of its enabledness zone, based
on the firing date sampled from any probabilistic distribution (we currently support seven continuous distributions and two discrete ones, but the approach is not limited only to these as we also support custom user-provided distributions). This allows
us to model a wide range of stochastic systems, as demonstrated in our case studies. 

We discussed the design and implementation of the statistical model checking algorithm
and argued that it is well-defined. All algorithms are implemented in the open-source model checker
TAPAAL, including numerous performance optimizations (e.g. parallel execution) as well as
a support for the visualization of the statistical model checking results. Our experiments
indicate that the formalism is applicable to a broad range of problems and allows us to reason
about complex stochastic behaviours in an intuitive way.

In the future work, we would like to extend the SMC reachability formulae to a more powerful 
logic (like e.g. LTL) that will allow us to reason about probability of runs that satisfy
additional requirements. An optimization of our SMC algorithm for handling rare events is another line
of research.

\paragraph{Acknowledgements.}
We thank Mikkel Tygesen for his help with TAPAAL GUI and in particular the visualization of the SMC results, and to Nikolaj Rossander Kristensen for his help with the generation of the custom Pauli-twirl distribution in the last case study.
This paper was funded by the Villum Investigator project S4OS and the ANR project BisoUS ANR-22-CE48-0012.

\bibliographystyle{fundam}
\bibliography{TAPN_SMC.bib}
\end{document}